\documentclass[11pt]{article}

\newlength{\actualtopmargin}
\newlength{\actualsidemargin}
\usepackage[tbtags]{amsmath}
\usepackage{amsthm}
\usepackage{amssymb}
\usepackage{amsfonts}
\usepackage{hyperref}

\theoremstyle{plain}
  \newtheorem{theorem}{Theorem}
  \newtheorem{lemma}[theorem]{Lemma}
  \newtheorem{corollary}[theorem]{Corollary}
  \newtheorem{proposition}[theorem]{Proposition}
  \newtheorem{claim}{Claim}
\theoremstyle{definition}
  \newtheorem{definition}[theorem]{Definition}

\theoremstyle{remark}
  \newtheorem*{remark}{Remark}
\theoremstyle{plain}
  \newtheorem*{theorem*}{Theorem}
  \newtheorem*{lemma*}{Lemma}
  \newtheorem*{corollary*}{Corollary}
  \newtheorem*{proposition*}{Proposition}
  \newtheorem*{claim*}{Claim}

\newenvironment{step}
  {
    \begin{enumerate}

  }
  {\end{enumerate}}

\newenvironment{algorithm*}[1]
  {
    \begin{center}
      \hrulefill\\
      \textbf{#1}
  }
  {
    \vspace{-1\baselineskip}
    \hrulefill
    \end{center}
  }

\newenvironment{protocol*}[1]
  {
    \begin{center}
      \hrulefill\\
      \textbf{#1}
  }
  {
    \vspace{-1\baselineskip}
    \hrulefill
    \end{center}
  }

\newlength{\itemwidth}
\newlength{\descriptionwidth}
\newenvironment{promiseproblem*}[4]
  {
    \begin{center}
      \hrulefill\\
      \textbf{\textsc{#1}}
      \settowidth{\itemwidth}{\textbf{Yes Instances:}}
      \setlength{\descriptionwidth}{\textwidth}
      \addtolength{\descriptionwidth}{-\itemwidth}
      \addtolength{\descriptionwidth}{-\labelsep}
      \begin{description}
        \item[\parbox{\itemwidth}{Input:}]
          \parbox[t]{\descriptionwidth}{#2}
        \item[\parbox{\itemwidth}{Yes Instances:}]
          \parbox[t]{\descriptionwidth}{#3}
        \item[\parbox{\itemwidth}{No Instances:}]
          \parbox[t]{\descriptionwidth}{#4}
      \end{description}
  }
  {
    \vspace{-1\baselineskip}
    \hrulefill
    \end{center}
  }

\newcommand{\bbC}{\mathbb{C}}

\newcommand{\bbN}{\mathbb{N}}

\newcommand{\bbZ}{\mathbb{Z}}

\newcommand{\bfC}{\mathbf{C}}
\newcommand{\bfD}{\mathbf{D}}

\newcommand{\bfL}{\mathbf{L}}

\newcommand{\bfT}{\mathbf{T}}
\newcommand{\bfU}{\mathbf{U}}

\newcommand{\calH}{\mathcal{H}}

\newcommand{\calK}{\mathcal{K}}

\newcommand{\calW}{\mathcal{W}}

\newcommand{\rmc}{\mathrm{c}}
\newcommand{\rmq}{\mathrm{q}}

\newcommand{\sfA}{\mathsf{A}}

\newcommand{\sfM}{\mathsf{M}}

\newcommand{\sfP}{\mathsf{P}}

\newcommand{\sfR}{\mathsf{R}}
\newcommand{\sfS}{\mathsf{S}}

\newcommand{\sfV}{\mathsf{V}}
\newcommand{\sfW}{\mathsf{W}}

\newcommand{\sfY}{\mathsf{Y}}
\newcommand{\sfZ}{\mathsf{Z}}

\newcommand{\classfont}{\mathrm}
\newcommand{\problemfont}{\textsc}

\newcommand{\BPP}{\classfont{BPP}}

\newcommand{\PP}{\classfont{PP}}

\newcommand{\PSPACE}{\classfont{PSPACE}}

\newcommand{\IP}{\classfont{IP}}

\newcommand{\QIP}{\classfont{QIP}}

\newcommand{\AM}{\classfont{AM}}
\newcommand{\QMA}{\classfont{QMA}}

\newcommand{\QCMA}{\classfont{QCMA}}
\newcommand{\MQA}{\classfont{MQA}}

\newcommand{\EPRQMA}[1]{\QMA^{#1\text{-}\EPR}}
\newcommand{\QAM}{\classfont{QAM}}
\newcommand{\GQAM}[1]{#1\textrm{-}\QAM}
\newcommand{\qqQAM}{\GQAM{\rmq\rmq}}
\newcommand{\qcQAM}{\GQAM{\rmq\rmc}}
\newcommand{\cqQAM}{\GQAM{\rmc\rmq}}
\newcommand{\ccQAM}{\GQAM{\rmc\rmc}}

\newcommand{\QMAM}{\classfont{QMAM}}

\newcommand{\QSZK}{\classfont{QSZK}}

\newcommand{\NIQSZK}{\classfont{NIQSZK}}

\newcommand{\BP}{\classfont{BP}}

\newcommand{\textEPRQMA}[1]{\textrm{QMA}^{#1\text{-}\EPR}}

\newcommand{\CITM}{\problemfont{CITM}}
\newcommand{\MaxOutQEA}{\problemfont{MaxOutQEA}}
\newcommand{\QSCTM}{\problemfont{QSCTM}}
\newcommand{\QEA}{\problemfont{QEA}}

\newcommand{\expectation}[1]{\mathbf{E}[ #1 ]}

\newcommand{\SD}{\mathrm{SD}}
\newcommand{\relent}[2]{D( #1 \mathbin{\Vert} #2 )}

\newcommand{\bra}[1]{\langle #1 \rvert}

\newcommand{\ket}[1]{\lvert #1 \rangle}

\newcommand{\ketbra}[1]{\lvert #1 \rangle \langle #1 \rvert}

\newcommand{\braket}[2]{\langle #1 \vert #2 \rangle}

\newcommand{\conjugate}[1]{{#1^\dagger}}

\newcommand{\tr}{\operatorname{tr}}

\newcommand{\tensor}{\otimes}
\newcommand{\norm}[1]{\lVert #1 \rVert}

\newcommand{\bignorm}[1]{\bigl\lVert #1 \bigr\rVert}

\newcommand{\trnorm}[1]{\lVert #1 \Vert_{\tr}}

\newcommand{\bigtrnorm}[1]{\bigl\lVert #1 \bigr\rVert_{\tr}}

\newcommand{\dnorm}[1]{\lVert #1 \rVert_{\diamond}}

\newcommand{\abs}[1]{\lvert #1 \rvert}

\newcommand{\bigabs}[1]{\bigl\lvert #1 \bigr\rvert}

\newcommand{\ceil}[1]{\lceil #1 \rceil}

\newcommand{\bigceil}[1]{\bigl\lceil #1 \bigr\rceil}

\newcommand{\floor}[1]{\lfloor #1 \rfloor}

\newcommand{\function}[3]{{#1 \colon #2 \to #3}}
\newcommand{\set}[2]{{\{ #1 \colon #2 \}}}

\newcommand{\bigset}[2]{{\bigl\{ #1 \colon #2 \bigr\}}}

\newcommand{\Complex}{\bbC}

\newcommand{\Natural}{\bbN}
\newcommand{\Integers}{\bbZ}
\newcommand{\Nonnegative}{{\Integers^+}}

\newcommand{\Binary}{{\{ 0, 1 \}}}
\newcommand{\Linear}{\bfL}

\newcommand{\Density}{\bfD}
\newcommand{\Unitary}{\bfU}
\newcommand{\Transform}{\bfT}
\newcommand{\Channel}{\bfC}

\newcommand{\poly}{\mathrm{poly}}
\newcommand{\const}{\mathrm{const}}

\newcommand{\xor}{\oplus}

\newcommand{\textmax}{\mathrm{max}}
\newcommand{\textmin}{\mathrm{min}}

\newcommand{\legal}{\mathrm{legal}}

\newcommand{\acc}{\mathrm{acc}}

\newcommand{\inp}{\mathrm{in}}
\newcommand{\out}{\mathrm{out}}
\newcommand{\yes}{\mathrm{yes}}
\newcommand{\no}{\mathrm{no}}

\newcommand{\all}{\mathrm{all}}

\newcommand{\EPR}{\mathrm{EPR}}

\newcommand{\ignore}[1]{}

\begin{document}

\sloppy

% ---------------------------------------------------------------------------
%   Title page
% ---------------------------------------------------------------------------

\title{\Large
  \textbf{
    Generalized Quantum Arthur-Merlin Games
  }
}

\author{
  Hirotada Kobayashi\footnotemark[1]\\
  \and
  Fran\c{c}ois Le Gall\footnotemark[2]\\
  \and
  Harumichi Nishimura\footnotemark[3]
}

\date{}

\maketitle
\thispagestyle{empty}
\pagestyle{plain}
\setcounter{page}{0}

\renewcommand{\thefootnote}{\fnsymbol{footnote}}

\vspace{-5mm}

\begin{center}
{\large
  \footnotemark[1]%
  Principles of Informatics Research Division\\
  National Institute of Informatics\\
  Tokyo, Japan\\
%%   2-1-2 Hitotsubashi, Chiyoda, Tokyo 101-8430, Japan\\
  [2.5mm]
  \footnotemark[2]%
  Department of Computer Science\\
  Graduate School of Information Science and Technology\\
  The University of Tokyo\\
  Tokyo, Japan\\
%%   7-3-1 Hongo, Bunkyo, Tokyo 113-0033, Japan\\
  [2.5mm]
  \footnotemark[3]%
  Department of Computer Science and Mathematical Informatics\\
  Graduate School of Information Science\\
  Nagoya University \\
  Nagoya, Aichi, Japan
%%   Furo-cho, Chikusa, Nagoya, Aichi 464-8601, Japan  
}\\
[5mm]
{\large 12 September 2014}\\
[8mm]
\end{center}

\renewcommand{\thefootnote}{\arabic{footnote}}

% ---------------------------------------------------------------------------
%   Abstract
% ---------------------------------------------------------------------------

\begin{abstract}
This paper investigates the role of interaction and coins in \emph{public-coin quantum interactive proof systems}
(also called \emph{quantum Arthur-Merlin games}).
While prior works focused on classical public coins even in the quantum setting,
the present work introduces a generalized version of quantum Arthur-Merlin games
where the public coins can be quantum as well:
the verifier can send not only random bits, but also halves of EPR pairs.
This generalization turns out to provide several novel characterizations
of constant-turn interactive proof systems.
First, it is proved that the class of two-turn quantum Arthur-Merlin games
with quantum public coins, denoted $\qqQAM$ in this paper,
does not change by adding a constant number of turns of classical interactions
prior to the communications of the $\mathrm{qq}$-QAM proof systems.
This can be viewed as a quantum analogue of the celebrated collapse theorem for $\AM$ due to Babai.
To prove this collapse theorem,
this paper provides a natural complete problem for $\qqQAM$:
deciding whether the output of a given quantum circuit is close to a totally mixed state.
This complete problem is on the very line of the previous studies investigating the hardness of checking the properties related to quantum circuits,
and is of independent interest.
It is further proved that the class~$\qqQAM_1$ of two-turn quantum-public-coin quantum Arthur-Merlin proof systems with perfect completeness
gives new bounds for standard well-studied classes of two-turn interactive proof systems.
Finally, the collapse theorem above is extended to comprehensively classify
the role of interaction and public coins in quantum Arthur-Merlin games:
it is proved that, for any constant~${m \geq 2}$, the class of problems
having an $m$-turn quantum Arthur-Merlin proof system is either equal to $\PSPACE$
or equal to the class of problems having a two-turn quantum Arthur-Merlin game of a specific type,
which provides a complete set of quantum analogues of Babai's collapse theorem.
\end{abstract}

\clearpage

%%% Main Part

% ---------------------------------------------------------------------------
%   Introduction
% ---------------------------------------------------------------------------

\section{Introduction}
\label{Section: introduction}

\paragraph{Background and motivation.}

Interactive proof systems~\cite{GolMicRac89SIComp, Bab85STOC}
play a central role in computational complexity
and has many applications such as probabilistic checkable proofs and zero-knowledge proofs.
The aim of such a system is the verification of an assertion
(e.g., verifying if an input is in a language)
by a party implementing a polynomial-time probabilistic computation,
called the verifier, interacting with another party with unlimited power,
called the prover, in polynomially many turns.
Two definitions are given on the secrecy of the coin which 
the verifier can flip:
Goldwasser, Micali, and Rackoff~\cite{GolMicRac89SIComp}
defined private-coin proof systems, where the prover cannot see the outcomes of coin flips,
while Babai~\cite{Bab85STOC} defined public-coin proof systems,
where the prover can see all the outcomes of coin flips.
Public-coin interactive proof systems are often called Arthur-Merlin games or Arthur-Merlin proof systems,
since in Ref.~\cite{Bab85STOC} the verifier was called Arthur
and the prover was called Merlin.

It is natural to expect that the power of interactive proof systems
depends on the number of interaction turns.
However, Babai~\cite{Bab85STOC} showed that as long as the number of turns is a constant at least two, 
the number of turns does not affect the power of Arthur-Merlin proof systems,
i.e., ${\AM(m) = \AM(2)}$ for any constant~${m \geq 2}$ (the \emph{collapse theorem}),
where ${\AM(m)}$ is the class of problems having an $m$-turn Arthur-Merlin proof system.
Goldwasser~and~Sipser~\cite{GolSip89ACR} then showed that 
a private-coin interactive proof system can be simulated
by an Arthur-Merlin proof system by adding two turns,
and thus, these two types of interactive proof systems are computationally equivalent.
By the above results, the class of problems having an interactive proof system
of a constant number of turns is equal to ${\AM(2)}$
(regardless of definitions with public coins or private coins),
and this class is nowadays called $\AM$.
The class~$\AM$ is believed to be much smaller than $\PSPACE$,
as it is contained in $\Pi_2^p$ in the second-level polynomial hierarchy~\cite{Lau83IPL, Bab85STOC}.
On the contrary, the class of problems having a more general interactive proof system of polynomially many turns,
called $\IP$, does coincide with $\PSPACE$~\cite{Pap85JCSS, LunForKarNis92JACM, Sha92JACM}
(again regardless of definitions with public coins or private coins~\cite{GolSip89ACR, She92JACM}).

Quantum interactive proof systems were introduced by Watrous~\cite{Wat03TCS},
and the class of problems having a quantum interactive proof system is called $\QIP$.
In the quantum world, the importance of the number of turns in interactive proof systems is drastically changed.
The first paper on quantum interactive proofs~\cite{Wat03TCS}
already proved the surprising power of constant-turn quantum interactive proof systems,
by showing that any problem in $\PSPACE$ has a three-turn quantum interactive proof system.
Kitaev~and~Watrous~\cite{KitWat00STOC} then proved that any quantum interactive proof system 
can be simulated by a three-turn quantum interactive proof system,
namely, ${\QIP = \QIP(3)}$,
where ${\QIP(m)}$ denotes the class of problems having an $m$-turn quantum interactive proof system.
Finally, the recent result~${\QIP = \PSPACE}$ by Jain,~Ji,~Upadhyay,~and~Watrous~\cite{JaiJiUpaWat11JACM}
completely characterized the computational power of
quantum interactive proof systems with three turns or more.
In contrast, despite of a number of intensive studies~\cite{Wat02FOCS, Weh06STACS, JaiUpaWat09FOCS, HayMilWil13CCC},
still very little is known on the class~$\QIP(2)$
corresponding to \emph{two-turn} quantum interactive proof systems,
and characterizing the computational power of two-turn quantum interactive proof systems
is one of the main open problems in this field.

A public-coin version of quantum interactive proof systems
was first introduced by Marriott~and~Watrous~\cite{MarWat05CC},
named quantum Arthur-Merlin proof systems,
where the messages from the verifier are restricted
to classical strings consisting only of outcomes of polynomially many attempts of a fair coin flip.
They then showed that three-turn quantum Arthur-Merlin proof systems
can simulate three-turn standard quantum interactive proof systems,
and hence the corresponding class, denoted $\QMAM$, coincides with ${\QIP = \PSPACE}$.
They also investigated the case of two-turn quantum Arthur-Merlin proof systems
and showed that the corresponding class, denoted $\QAM$, is included in ${\BP \cdot \PP}$,
a subclass of $\PSPACE$ obtained by applying the $\BP$~operator to the class~$\PP$,
which is still the only nontrivial upper bound known for $\QAM$.

\paragraph{Results and their meanings.}

This paper introduces a ``quantum public-coin'' version of quantum interactive proof systems,
which generalizes quantum Arthur-Merlin proof systems in Ref.~\cite{MarWat05CC}.
In this generalized model, the verifier can send quantum messages, but these messages can be only used for sharing EPR pairs with the prover,
i.e., the verifier at his/her turn first generates polynomially many EPR pairs
and then sends one half of each of them to the prover.
The main interest in this model is again on the two-turn case,
as allowing three or more turns in this model obviously hits the $\PSPACE$ ceiling.
Let $\qqQAM$ be the class of problems
having a two-turn ``quantum public-coin'' interactive proof system
in which the first message from the verifier consists only of polynomially many halves of EPR pairs.
Note that the only difference from the existing class~$\QAM$ 
lies in the type of the message from the verifier:
uniform random classical bits are replaced by halves of EPR pairs,
which can be thought as a natural quantum version of classical public coins.
The main goal of this paper is to investigate the computational power of this class~$\qqQAM$
in order to figure out the advantages offered by quantum public-coins,
and more generally, to make a step forward in the understanding of
two-turn quantum interactive proof systems.

While the class~$\qqQAM$ is the main target of investigation,
this paper further studies the power of various models of quantum Arthur-Merlin proofs with quantum/classical public coins.
For any constant~${m \geq 1}$ and any ${t_1, \ldots, t_m}$ in ${\{\rmc, \rmq\}}$,
let ${\GQAM{t_m \cdots t_1}(m)}$ be the class of problems
that have an $m$-turn quantum interactive proof system with the following restrictions:
\begin{itemize}
\item
  For any odd~$j$, ${1 \leq j \leq m}$,
  the ${(m-j+1)}$st message (or the $j$th message counting from the last),
  which is the message from the prover sent at the ${(m-j+1)}$st turn,
  is a quantum message if ${t_j = \rmq}$,
  and is restricted to a classical message if ${t_j = \rmc}$.
\item
  For any even~$j$, ${1 \leq j \leq m}$,
  at the ${(m-j+1)}$st turn, which is a turn for the verifier,
  the verifier first generates polynomially many EPR pairs and then sends halves of them if ${t_j = \rmq}$,
  while the verifier flips a fair coin polynomially many times and then sends their outcomes if ${t_j = \rmc}$.
\end{itemize}
The class~${\GQAM{t_m\cdots t_1}(m)}$ may be simply written as $\GQAM{t_m \cdots t_1}$
when there is no ambiguity in the number of turns: 
for instance, ${\qqQAM(2)}$ may be abbreviated to $\qqQAM$.
Note that the classes~$\QAM$~and~$\QMAM$ defined in Ref.~\cite{MarWat05CC}
are exactly the classes~$\cqQAM$ and $\GQAM{\mathrm{qcq}}$, respectively.
The class~$\ccQAM$ corresponds to two-turn public-coin quantum interactive proofs with classical communications:
the verifier sends a question consisting only of outcomes of polynomially many 
attempts of a fair coin flip,
then the prover responds with polynomially many classical bits,
and the final verification is done by the verifier via polynomial-time quantum computation.
By definition,
${\AM \subseteq \ccQAM \subseteq \cqQAM \subseteq \qqQAM \subseteq \QIP(2)}$.

As mentioned above, the main target in this paper is the class $\qqQAM$.
First, it is proved that the power of $\mathrm{qq}$-QAM proof systems
does not change by adding a constant number of turns of classical interactions
prior to the communications of the $\mathrm{qq}$-QAM proof systems.

\begin{theorem}
For any constant~${m \geq 2}$,
${
\mathrm{c \cdots c}\qqQAM(m) = \qqQAM
}$. 
\label{Theorem: c...cqq-QAM(m) = qq-QAM}
\end{theorem}

In stark contrast to this,
as mentioned before
and will be stated clearly in Theorem~\ref{Theorem: complete classifications},
adding one turn of prior quantum interaction gives the $\mathrm{qq}$-QAM proof systems
the full power of quantum interactive proof systems
(i.e., the resulting class is $\PSPACE$).
Hence, Theorem~\ref{Theorem: c...cqq-QAM(m) = qq-QAM} may be viewed as a quantum analogue of
Babai's collapse theorem~\cite{Bab85STOC} for the class~$\qqQAM$.

The proof of Theorem~\ref{Theorem: c...cqq-QAM(m) = qq-QAM} comes in three parts:
The first part proves that,
for any constant~${m \geq 4}$,
${\mathrm{c \cdots c}\qqQAM(m)}$ is necessarily included in ${\mathrm{cc}\qqQAM}$.
The second part proves that
${\mathrm{c}\qqQAM}$ is included in $\qqQAM$.
Finally, the third part proves that
${\mathrm{cc}\qqQAM}$ is included in $\qqQAM$,
by using the containment proved in the second part.

The first part is proved by carefully extending the argument in Babai's collapse theorem.
The core idea of Babai's proof is that,
by a probabilistic argument applied to a parallel repetition of the original proof system,
the order of the verifier and the prover in the first three turns
of the original system can be switched,
which results in another proof system that has fewer number of turns.
When proving the first part
the messages of the first three turns of the original $m$-turn QAM proof system are classical,
and thus, the argument in Babai's collapse theorem still works.

The proof of the second part is one of the highlights in this paper.
The main difficulty in proving this part (and the third part) is that
the argument used in Babai's collapse theorem fails
when any of the first three turns is quantum in the starting proof system.

To overcome this difficulty,
this paper first provides a natural complete promise problem for $\qqQAM$,
namely,
the \problemfont{Close Image to Totally Mixed} ($\CITM$) problem,
which asks to check if the image of a given quantum circuit
can be close to a totally mixed state, formally defined as follows.

\begin{promiseproblem*}
  {Close Image to Totally Mixed Problem: $\boldsymbol{\CITM(a,b)}$}
  {
    A description of a quantum circuit~$Q$ acting on $q_\all$~qubits
    that has $q_\inp$~specified input qubits and $q_\out$~specified output qubits.
  }
  {
    There exists a quantum state~$\rho$ of $q_\inp$~qubits such that
    ${D(Q(\rho), (I/2)^{\tensor q_\out}) \leq a}$.
  }
  {
    For any quantum state~$\rho$ of $q_\inp$~qubits,
    ${D(Q(\rho), (I/2)^{\tensor q_\out}) \geq b}$.
  }
\end{promiseproblem*}
Here, ${D(\cdot, \cdot)}$ denotes the trace distance,
${Q(\rho)}$ is the $q_\out$-qubit output state of $Q$ when the input state was $\rho$
(i.e., the reduced state obtained by tracing out
the space corresponding to the ${(q_\all - q_\out)}$~non-output qubits
after applying $Q$ to ${\rho \tensor (\ketbra{0})^{\tensor (q_\all - q_\inp)}}$),
and $I$ is the identity operator of dimension two
(and thus, ${(I/2)^{\tensor q_\out}}$ corresponds to the totally mixed state of $q_\out$~qubits).
The following completeness result is proved.

\begin{theorem}
For any constants~$a$~and~$b$ in ${(0,1)}$ such that ${(1 - a)^2 > 1 - b^2}$,
${\CITM(a,b)}$ is $\qqQAM$-complete under polynomial-time many-one reduction.
\label{Theorem: qq-QAM-completeness of CITM}
\end{theorem}

Then the core idea for proving the second part is
to use the structure of this complete problem that
yes-instances are witnessed by the existence of a quantum state
(i.e., the $\exists$ quantifier appears in the first place),
while no such witness quantum state exists for no-instances
(i.e., the $\forall$ quantifier appears in the first place).
This makes it possible to incorporate the first turn of the $\mathrm{cqq}$-QAM system
into the input quantum state of the complete problem $\CITM$
(as the quantifier derived from the first turn of the $\mathrm{cqq}$-QAM system
matches the quantifier derived from the complete problem $\CITM$),
and thus, any problem in ${\mathrm{c}\qqQAM}$ can be reduced in polynomial time
to the $\CITM$ problem with appropriate parameters, which is in $\qqQAM$.

Actually, for the proof, 
whether the image of a constructed quantum circuit can be close to a totally mixed state
is partly evaluated by using the maximum output entropy of quantum channels, 
which shows implicitly the $\qqQAM$-completeness of another problem
that asks to check whether the maximum output entropy of a quantum channel is larger than a given value or not.
More formally, the following \problemfont{Maximum Output Quantum Entropy Approximation} ($\MaxOutQEA$) problem is also $\qqQAM$-complete.

\begin{promiseproblem*}
  {Maximum Output Quantum Entropy Approximation Problem: $\boldsymbol{\MaxOutQEA}$}
  {
    A description of a quantum circuit that specifies a quantum channel~$\Phi$,
    and a positive integer~$t$.
  }
  {
    ${S_\textmax(\Phi) \geq t + 1}$.
  }
  {
    ${S_\textmax(\Phi) \leq t - 1}$.
  }
\end{promiseproblem*}
Here, ${S_\textmax(\cdot)}$ denotes the maximum output von~Neumann entropy.
Namely, ${S_\textmax(\Phi) = \max_\rho S(\Phi(\rho))}$,
where ${S(\cdot)}$ denotes the von~Neumann entropy
and ${\Phi(\rho)}$ is the output quantum state of the quantum channel~$\Phi$
when the input quantum state to it was $\rho$.

\begin{theorem}
$\MaxOutQEA$ is $\qqQAM$-complete under polynomial-time many-one reduction.
\label{Theorem: qq-QAM-completeness of MaxOutQEA}
\end{theorem}

Finally, the third part then can be proved
by first providing a randomized reduction from a problem in ${\mathrm{cc}\qqQAM}$
to a problem in ${\mathrm{c}\qqQAM}$,
and then using the containment proved in the second part for the resulting problem in ${\mathrm{c}\qqQAM}$.

Besides its usefulness in proving Theorem~\ref{Theorem: c...cqq-QAM(m) = qq-QAM},
the complete problem $\CITM$ is of independent interest in the following sense.
Recall that problems with formulations similar to $\CITM$ have already been studied,
and were crucial to understand and characterize the computational power of several classes related to quantum interactive proof systems:
testing closeness between the images of two given quantum circuits is $\QIP$-complete~\cite{RosWat05CCC} (and hence $\PSPACE$-complete),
testing closeness between a state produced by a given circuit and the image of another quantum circuit is ${\QIP(2)}$-complete~\cite{Wat02QIP}
(see also Ref.~\cite{HayMilWil12arXiv}),
testing closeness between two states produced by two given quantum circuits is $\QSZK$-complete~\cite{Wat02FOCS, Wat09SIComp},
and testing closeness between the state produced by a quantum circuit and the totally mixed state is $\NIQSZK$-complete~\cite{Kob03ISAAC, ChaCioKerVad08TCC}.
Theorem~\ref{Theorem: qq-QAM-completeness of CITM} shows that the class~$\qqQAM$,
besides its theoretical interest in the context of interactive proofs,
is a very natural one that actually corresponds to a concrete computational problem
that is on this line of studies investigating the hardness of checking the properties related to quantum circuits.
Since $\CITM$ corresponds to the remaining pattern (image versus totally mixed state),
Theorem~\ref{Theorem: qq-QAM-completeness of CITM}
provides the last piece for characterizing the hardness of these kinds of computational problems.

It is further proved that the class~$\cqQAM$ (i.e., the standard $\QAM$)
is necessarily contained in the one-sided bounded error version of $\qqQAM$ of perfect completeness,
denoted by $\qqQAM_1$
(throughout this paper,
the perfect completeness version of each complexity class
is indicated by adding the subscript~``$1$'').

\begin{theorem}
${\cqQAM \subseteq \qqQAM_1}$.
\label{Theorem: QAM is in qq-QAM_1}
\end{theorem}

One useful property when proving this theorem is that
the proof of Theorem~\ref{Theorem: c...cqq-QAM(m) = qq-QAM}
does not harm the perfect completeness property,
i.e., it also holds that ${\mathrm{c \cdots c}\qqQAM_1(m) = \qqQAM_1}$,
for any constant~${m \geq 2}$.
Especially,
the class~${\mathrm{cc}\qqQAM_1}$ is included in the class~$\qqQAM_1$,
and thus,
one has only to prove that $\cqQAM$ is included in ${\mathrm{cc}\qqQAM_1}$.
This can be proved by combining the classical technique
due to Cai~\cite{Cai12LectureNotes} for proving ${\AM = \AM_1}$
(which itself originates in the proof of ${\BPP \subseteq \Sigma_2^p}$
due to Lautemann~\cite{Lau83IPL}),
and the recent result that
any problem in $\QMA$ has a one-sided bounded error QMA system of perfect completeness
in which Arthur and Merlin initially share a constant number of EPR pairs~\cite{KobLeGNis13ITCS}
(which in particular implies that $\QMA$ is included in $\qqQAM_1$).
Now the point is that, using two classical turns,
the classical technique in Ref.~\cite{Cai12LectureNotes}
can be used to generate polynomially many instances of a (promise) QMA problem,
all of which are QMA yes-instances if the input was a yes-instance,
while at least one of which is a QMA no-instance with high probability
if the input was a no-instance.
Hence, by making use of the proof system in Ref.~\cite{KobLeGNis13ITCS}
for each QMA instance,
which essentially runs polynomially many attempts of
a protocol of $\mathrm{qq}$-QAM type in parallel 
to check that none of them results in rejection,
one obtains a proof system of $\mathrm{ccqq}$-QAM type with perfect completeness.

An immediate corollary of this theorem is the first nontrivial upper bound for $\QAM$
in terms of quantum interactive proofs.

\begin{corollary}
${\QAM \subseteq \QIP_1(2)}$.
\end{corollary}

Here, ${\QIP_1(2)}$ denotes the class of problems having a two-turn quantum interactive proof system of perfect completeness.
This also improves the best known lower bound of ${\QIP_1(2)}$
(from $\QMA$ shown in Ref.~\cite{KobLeGNis13ITCS} to $\QAM$).
By using the fact~${\MQA = \MQA_1}$ (a.k.a., ${\QCMA = \QCMA_1}$)
stating that classical-witness QMA systems can be made perfectly complete~\cite{JorKobNagNis12QIC},
a technique similar to the proof of Theorem~\ref{Theorem: QAM is in qq-QAM_1}
proves that perfect completeness is achievable in $\ccQAM$.

\begin{theorem}
${\ccQAM = \ccQAM_1}$.
\label{Theorem: ccQAM = ccQAM_1}
\end{theorem}

Finally, results similar to Theorem~\ref{Theorem: c...cqq-QAM(m) = qq-QAM} can be derived
for other complexity classes related to the generalized quantum Arthur-Merlin proof systems.
Namely, the following complete characterization is proved
on the power of constant-turn generalized quantum Arthur-Merlin proofs,
which can be viewed as the complete set of quantum analogues of Babai's collapse theorem.

\begin{theorem}
The following four properties hold:
\begin{itemize}
\item[(i)]
  For any constant~${m \geq 3}$ and any ${t_1, \ldots, t_m}$ in ${\{\rmc, \rmq\}}$,
  if there exists an index~${j \geq 3}$ such that ${t_j = \rmq}$, then
  ${\GQAM{t_m \cdots t_1}(m) = \PSPACE}$.
\item[(ii)]
  For any constant~${m \geq 2}$ and any $t_1$ in ${\{\rmc, \rmq\}}$,
  ${\GQAM{\mathrm{c \cdots cq} \:\! t_1}(m) = \qqQAM}$.
\item[(iii)]
  For any constant~${m \geq 2}$,
  ${\GQAM{\mathrm{c \cdots cq}}(m) = \cqQAM}$ (${= \QAM}$).
\item[(iv)]
  For any constant~${m \geq 2}$,
  ${\GQAM{\mathrm{c \cdots c}}(m) = \ccQAM}$.
\end{itemize}
\label{Theorem: complete classifications}
\end{theorem}

\paragraph{Further related work.}

There are several works in which relevant subclasses of $\qqQAM$ were treated.
In Ref.~\cite{KobLeGNis13ITCS}, the class~$\EPRQMA{\const}$ was introduced
to give an upper bound of $\QMA$ by its one-sided bounded error subclass~$\EPRQMA{\const}_1$
with perfect completeness.
This $\EPRQMA{\const}$ is an obvious subclass of $\qqQAM$
with a restriction that the first message from the verifier consists of
not polynomially many but a constant number of halves of EPR pairs.
The class~$\qqQAM$ may be called $\EPRQMA{\poly}$,
following the notation in Ref.~\cite{KobLeGNis13ITCS}.
Another subclass of $\qqQAM$ is the class~$\NIQSZK$ studied in Refs.~\cite{Kob03ISAAC, ChaCioKerVad08TCC}
that corresponds to non-interactive quantum statistical zero-knowledge proof systems,
where the zero-knowledge property must also be satisfied.

\paragraph{Organization of the paper.}

Section~\ref{Section: preliminaries} summarizes the notions and properties
that are used throughout this paper,
and gives formal definitions of generalized quantum Arthur-Merlin proof systems.
Section~\ref{Section: qq-QAM-completeness of CITM}
proves the $\qqQAM$-completeness of the $\CITM$ problem.
Section~\ref{Section: c...cqq-QAM(m) = qq-QAM}
then gives a proof of Theorem~\ref{Theorem: c...cqq-QAM(m) = qq-QAM},
the collapse theorem for $\qqQAM$.
This essentially proves the $\qqQAM$-completeness of the $\MaxOutQEA$ problem also.
Section~\ref{Section: QAM is in qq-QAM_1}
treats the result that the standard $\QAM$ is contained in $\qqQAM_1$,
the perfect-completeness version of $\qqQAM$.
Section~\ref{Section: collapse theorem for general QAM}
presents the complete classification of the complexity classes
derived from generalized quantum Arthur-Merlin proof systems.
Finally, Section~\ref{Section: conclusion}
concludes the paper with some open problems.
For completeness,
a rigorous proof of the $\qqQAM$-completeness of the $\MaxOutQEA$ problem (Theorem~\ref{Theorem: qq-QAM-completeness of MaxOutQEA}) is given in the Appendix.

% ---------------------------------------------------------------------------
%   Preliminaries
% ---------------------------------------------------------------------------

\section{Preliminaries}
\label{Section: preliminaries}

Throughout this paper,
let $\Natural$ and $\Nonnegative$ denote
the sets of positive and nonnegative integers, respectively,
and let ${\Sigma = \Binary}$ denote the binary alphabet set. 
A function~$\function{f}{\Nonnegative}{\Natural}$ is \emph{polynomially bounded}
if there exists a polynomial-time deterministic Turing machine
that outputs ${1^{f(n)}}$ on input~$1^n$.
A function~$\function{f}{\Nonnegative}{[0,1]}$ is \emph{negligible}
if, for every polynomially bounded function~$\function{g}{\Nonnegative}{\Natural}$,
it holds that ${f(n) < 1/g(n)}$ for all but finitely many values of~$n$.

% ---------------------------------------------------------------------------
%   Quantum Fundamentals
% ---------------------------------------------------------------------------

\subsection{Quantum Fundamentals}
\label{Subsection: quantum fundamentals}

We assume the reader is familiar with the quantum formalism,
including pure and mixed quantum states, density operators,
measurements, trace norm, fidelity, as well as the quantum circuit model
(see Refs.~\cite{NieChu00Book, KitSheVya02Book}, for instance).
This subsection summarizes some notations and properties that are used in this paper.

For each $k$ in $\Natural$,
let ${\Complex(\Sigma^k)}$ denote the $2^k$-dimensional complex Hilbert space
whose standard basis vectors are indexed by the elements in $\Sigma^k$.
In this paper, all Hilbert spaces are complex and have dimension a power of two.
For a Hilbert space~$\calH$,
let $I_\calH$ denote the identity operator over $\calH$,
and let ${\Density(\calH)}$ and ${\Unitary(\calH)}$ be the sets of density and unitary operators over $\calH$, respectively.
For a quantum register~$\sfR$,
let $\ket{0}_\sfR$ denote the state
in which all the qubits in $\sfR$ are in state~$\ket{0}$.
As usual,
denote the four two-qubit states in ${\Complex(\Sigma^2)}$
that form the \emph{Bell basis} by
\begin{alignat*}{4}
&
\ket{\Phi^+}
&&
=
\frac{1}{\sqrt{2}} (\ket{00} + \ket{11}),
&
\quad
&
\ket{\Phi^-}
&&
=
\frac{1}{\sqrt{2}} (\ket{00} - \ket{11}),
\\
&
\ket{\Psi^+}
&&
=
\frac{1}{\sqrt{2}} (\ket{01} + \ket{10}),
&
\quad
&
\ket{\Psi^-}
&&
=
\frac{1}{\sqrt{2}} (\ket{01} - \ket{10}),
\end{alignat*}
respectively.
Let
\[
  X
  =
  \begin{pmatrix}
    0 & 1\\
    1 & 0
  \end{pmatrix},
  \quad
  Z
  =
  \begin{pmatrix}
    1 & 0\\
    0 & -1
  \end{pmatrix}
\]
denote the Pauli operators.
For convenience, we may identify a unitary operator with the unitary transformation it induces.
In particular, for a unitary operator~$U$,
the induced unitary transformation is also denoted by $U$.

For two Hilbert spaces~$\calH$~and~$\calK$
and a quantum state~$\rho$ in ${\Density(\calH \tensor \calK)}$,
the state obtained from $\rho$ by \emph{tracing out} $\calK$
(i.e., discarding the qubits in the \emph{reference system} corresponding to $\calK$)
is the \emph{reduced state} in ${\Density(\calH)}$ of $\rho$
denoted by ${\tr_\calK \rho}$.
For two Hilbert spaces~$\calH$~and~$\calK$,
a pure quantum state~$\ket{\psi}$ in ${\calH \tensor \calK}$
is a \emph{purification} of a mixed quantum state~$\rho$ in ${\Density(\calH)}$
iff ${\tr_\calK \ketbra{\psi} = \rho}$.

For a linear operator~$A$, the \emph{trace norm} of $A$ is defined by
\[
\trnorm{A} = \tr \sqrt{\conjugate{A}A}.
\]
For two quantum states~$\rho$~and~$\sigma$,
the \emph{trace distance} between them is defined by
\[
D(\rho, \sigma) = \frac{1}{2} \trnorm{\rho - \sigma}.
\]
A special case of the trace distance is the \emph{statistical difference}
between two probability distributions~$\mu$~and~$\nu$,
which is defined by
\[
\SD(\mu, \nu) = D(\mu, \nu)
\]
by viewing probability distributions as special cases of quantum states
with diagonal density operators.
The following important property is well known
on probability distributions derived from quantum states.

\begin{lemma}
Let $\mu_\rho$ and $\mu_\sigma$ be the probability distributions
derived from two quantum states~$\rho$~and~$\sigma$, respectively,
by performing an arbitrary identical measurement.
Then,
\[
  \SD(\mu_\rho, \mu_\sigma) \leq D(\rho, \sigma).
\]
\label{Lemma: trace distance and probability}
\end{lemma}

For two quantum states~$\rho$~and~$\sigma$,
the \emph{fidelity} between them is defined by
\[
F(\rho, \sigma) = \tr \sqrt{\sqrt{\rho} \sigma \sqrt{\rho}}.
\]
In particular, for two pure states~$\ket{\phi}$~and~$\ket{\psi}$,
the fidelity between them is given by
${
F(\ketbra{\phi}, \ketbra{\psi}) = \abs{\braket{\phi}{\psi}}
}$.  
The fidelity can also be represented as follows~\cite{Uhl76RoMP}.

\begin{lemma}[Uhlmann's theorem]
For any Hilbert spaces~$\calH$~and~$\calK$ satisfying ${\dim \calK \geq \dim \calH}$
and any quantum states~$\rho$~and~$\sigma$ in ${\Density(\calH)}$,
let $\ket{\phi_\rho}$ and $\ket{\phi_\sigma}$ in ${\calH \tensor \calK}$
be any purifications of $\rho$ and $\sigma$.
Then,
\[
F(\rho, \sigma) = \max \bigset{\bigabs{\bra{\phi_\rho} (I_\calH \tensor U) \ket{\phi_\sigma}}}{U \in \Unitary(\calK)}.
\]
\label{Lemma: Uhlmann's theorem}
\end{lemma}

The following inequalities relate the trace distance and fidelity~\cite{FucGra99IEEEIT}.

\begin{lemma}[Fuchs-van-de-Graaf inequalities]
For any Hilbert space~$\calH$ and any quantum states~$\rho$~and~$\sigma$ in ${\Density(\calH)}$, 
\[
  1 - F(\rho, \sigma) \leq D(\rho, \sigma) \leq \sqrt{1 - (F(\rho, \sigma))^2}.
\]
\label{Lemma: trace distance and fidelity}
\end{lemma}

This paper also uses the following property.

\begin{lemma}
For any Hilbert space~$\calH$,
any quantum states~$\rho$,~$\sigma$,~and~$\tau$ in ${\Density(\calH)}$,
and any constant~$p$ in ${[0,1]}$, 
\[
D \bigl( (1-p) \rho + p \sigma, \tau) \geq D(\rho, \tau) - p.
\]
\label{Lemma: mixture of states and trace distance}
\end{lemma}

\begin{proof}
By the triangle inequality,
${
\trnorm{(1-p) \rho + p \sigma - \tau}
\geq
\trnorm{\rho - \tau}
-
p \trnorm{\rho - \sigma}
}$,
and thus,
\[
D \bigl( (1-p) \rho + p \sigma, \tau \bigr)
\geq
D(\rho, \tau) - p D(\rho, \sigma)
\geq
D(\rho, \tau) - p,
\]
as desired.
\end{proof}

For Hilbert spaces~$\calH$~and~$\calK$,
let ${\Linear(\calH)}$ denote the set of linear mappings from $\calH$ to itself,
let ${\Transform(\calH, \calK)}$ denote
the set of linear mappings from ${\Linear(\calH)}$ to ${\Linear(\calK)}$,
and let ${\Channel(\calH, \calK)}$ denote
the set of quantum channels from ${\Density(\calH)}$ to ${\Density(\calK)}$
(i.e., the set of linear mappings from ${\Linear(\calH)}$ to ${\Linear(\calK)}$
that are completely positive and trace-preserving).

For a linear mapping~$\Phi$ in ${\Transform(\calH, \calK)}$,
the \emph{diamond norm} of $\Phi$ is defined by
\[
\dnorm{\Phi} = \max \bigset{\trnorm{(\Phi \tensor I_{\Linear(\calH)}) (\rho)}}{\rho \in \Density(\calH^{\tensor 2})},
\]
where $I_{\Linear(\calH)}$ is the identity mapping over ${\Linear(\calH)}$.

For Hilbert spaces~$\calH$~and~$\calK$
and quantum channels~$\Phi$~and~$\Psi$ in ${\Channel(\calH, \calK)}$,
the \emph{minimum output trace distance} between $\Phi$ and $\Psi$ is defined by
\[
D_\textmin (\Phi, \Psi)
=
\min \set{D(\Phi(\rho), \Psi(\sigma))}{\rho, \sigma \in \Density(\calH)},
\]
and the \emph{maximum output fidelity} between $\Phi$ and $\Psi$ is defined by
\[
F_\textmax (\Phi, \Psi)
=
\max \set{F(\Phi(\rho), \Psi(\sigma))}{\rho, \sigma \in \Density(\calH)}.
\]

The Fuchs-van-de-Graaf inequalities relate
the minimum output trace distance and the maximum output fidelity as follows.

\begin{lemma}
For any Hilbert spaces~$\calH$~and~$\calK$
and any quantum channels~$\Phi$~and~$\Psi$ in ${\Channel(\calH, \calK)}$,
\[
  1 - F_\textmax (\Phi, \Psi) \leq D_\textmin (\Phi, \Psi) \leq \sqrt{1 - (F_\textmax (\Phi, \Psi))^2}.
\]
\label{Lemma: minimum output trace distance and maximum output fidelity}
\end{lemma}

\begin{proof}
Let $\rho_\ast$ and $\sigma_\ast$ be the quantum states in ${\Density(\calH)}$ 
that minimize the expression~${D(\Phi(\rho), \Psi(\sigma))}$.
Then,
\[
1 - F_\textmax (\Phi, \Psi)
\leq
1 - F(\Phi(\rho_\ast), \Psi(\sigma_\ast))
\leq
D(\Phi(\rho_\ast), \Psi(\sigma_\ast))
=
D_\textmin (\Phi, \Psi),
\]
and thus, the first inequality holds.
Similarly, let $\rho_\star$ and $\sigma_\star$ be the quantum states in ${\Density(\calH)}$
that maximize the expression~${F(\Phi(\rho), \Psi(\sigma))}$.
Then,
\[
D_\textmin (\Phi, \Psi)
\leq
D(\Phi(\rho_\star), \Psi(\sigma_\star))
\leq
\sqrt{1 - \bigl( F(\Phi(\rho_\star), \Psi(\sigma_\star)) \bigr)^2}
=
\sqrt{1 - (F_\textmax (\Phi, \Psi))^2},
\]
and the second inequality holds.
\end{proof}

The following property is implicit in Ref.~\cite{KitWat00STOC},
which can be proved by using the multiplicativity of the diamond norm
(see Problem~11.10 of Ref.~\cite{KitSheVya02Book} as well as Theorem~3.24 of Ref.~\cite{Ros09PhD}, for instance).

\begin{lemma}
For any Hilbert spaces~$\calH_1$,~$\calK_1$,~$\calH_2$,~and~$\calK_2$,
and any quantum channels~$\Phi_1$~and~$\Psi_1$ in ${\Channel(\calH_1, \calK_1)}$
and $\Phi_2$~and~$\Psi_2$ in ${\Channel(\calH_2, \calK_2)}$,
\[
F_\textmax (\Phi_1 \tensor \Phi_2, \Psi_1 \tensor \Psi_2)
=
F_\textmax (\Phi_1, \Psi_1) F_\textmax(\Phi_2, \Psi_2).
\] 
\label{Lemma: multiplicativity of maximum output fidelity}
\end{lemma}

From Lemmas~\ref{Lemma: minimum output trace distance and maximum output fidelity}~and~\ref{Lemma: multiplicativity of maximum output fidelity},
one can show the following.

\begin{lemma}
For any Hilbert spaces~$\calH$~and~$\calK$,
any quantum channels~$\Phi$~and~$\Psi$ in ${\Channel(\calH, \calK)}$,
and any $k$ in $\Natural$,
\[
1 - \bigl[ 1 - (D_\textmin (\Phi, \Psi))^2 \bigr]^{\frac{k}{2}}
\leq
D_\textmin (\Phi^{\tensor k}, \Psi^{\tensor k})
\leq
k D_\textmin (\Phi, \Psi).
\]
\label{Lemma: polarization of minimum output trace distance}
\end{lemma}

\begin{proof}
From Lemmas~\ref{Lemma: minimum output trace distance and maximum output fidelity}~and~\ref{Lemma: multiplicativity of maximum output fidelity},
it holds that
\[
1 - \bigl[ 1 - (D_\textmin (\Phi, \Psi))^2 \bigr]^{\frac{k}{2}}
\leq
1 - (F_\textmax (\Phi, \Psi))^k
=
1 - F_\textmax \bigl( \Phi^{\tensor k}, \Psi^{\tensor k} \bigr)
\leq
D_\textmin (\Phi^{\tensor k}, \Psi^{\tensor k}),
\]
and the first inequality of the claim follows.

On the other hand, by the triangle inequality,
for any quantum states~$\rho$~and~$\sigma$ in ${\Density(\calH)}$,
\[
\begin{split}
D \bigl( (\Phi(\rho))^{\tensor k}, (\Psi(\sigma))^{\tensor k} \bigr)
&
\leq
D \bigl( (\Phi(\rho))^{\tensor k}, \Psi(\sigma) \tensor (\Phi(\rho))^{\tensor (k-1)} \bigr)
+
D \bigl( \Psi(\sigma) \tensor (\Phi(\rho))^{\tensor (k-1)}, (\Psi(\sigma))^{\tensor k} \bigr)
\\
&
=
D \bigl( \Phi(\rho), \Psi(\sigma) \bigr)
+
D \bigl( (\Phi(\rho))^{\tensor (k-1)}, (\Psi(\sigma))^{\tensor (k-1)} \bigr).
\end{split}
\]
By repeatedly applying this bound with $\rho_\ast$ and $\sigma_\ast$ in ${\Density(\calH)}$
that minimize the expression~${D(\Phi(\rho), \Psi(\sigma))}$,
it holds that
\[
D_\textmin (\Phi^{\tensor k}, \Psi^{\tensor k})
\leq
D \bigl( (\Phi(\rho_\ast))^{\tensor k}, (\Psi(\sigma_\ast))^{\tensor k} \bigr)
\leq
k D \bigl( \Phi(\rho_\ast), \Psi(\sigma_\ast) \bigr)
=
k D_\textmin (\Phi, \Psi),
\]
and the second inequality of the claim follows.
\end{proof}

Finally,
for any quantum state~$\rho$,
the \emph{von~Neumann entropy} of $\rho$ is defined by
\[
S(\rho) = - \tr(\rho \log \rho).
\]
A special case of the von~Neumann entropy is the \emph{Shannon entropy}
of a probability distribution~$\mu$,
which is defined by
\[
H(\mu) = S(\mu)
\]
by viewing probability distributions as special cases of quantum states
with diagonal density operators.

For Hilbert spaces~$\calH$~and~$\calK$
and a quantum channel~$\Phi$ in ${\Channel(\calH, \calK)}$,
the \emph{maximum output von~Neumann entropy} of $\Phi$ is defined by
\[
S_\textmax(\Phi) = \max \set{S(\Phi(\rho))}{\rho \in \Density(\calH)}.
\]

This paper uses the following two properties on von~Neumann entropy.

The first lemma provides an upper bound on the von~Neumann entropy of a mixture 
of quantum states~\cite[Theorem 11.10]{NieChu00Book}.

\begin{lemma}
For any Hilbert space~$\calH$ and any quantum state~$\rho$ in ${\Density(\calH)}$
such that
${\rho = \sum_i \mu_i \rho_i}$
for some probability distribution~${\mu = \{ \mu_i \}}$
and quantum states~$\rho_i$ in ${\Density(\calH)}$,
\[
S(\rho) \leq H(\mu) + \sum_i \mu_i S(\rho_i).
\]
\label{Lemma: NC00Thm11.10}
\end{lemma}

The second lemma describes relations
between the von~Neumann entropy of a quantum state
and the trace distance between the state and the totally mixed state
(a similar but slightly stronger statement appeared in Ref.~\cite{ChaCioKerVad07CePrint} without a proof).

\begin{lemma}
For any quantum state~$\rho$ of $n$~qubits,
it holds that
\[
\bigl( 1 - D(\rho, (I/2)^{\tensor n}) - 2^{-n} \bigr) n
\leq
S(\rho)
\leq
n - \log \frac{1}{1 - D(\rho, (I/2)^{\tensor n})} + 2.
\]
\label{Lemma: trace distance and entropy}
\end{lemma}

\begin{proof}
First we show the first inequality. 
By considering the spectral decomposition of $\rho$,
one can write
${\rho = \sum_{x \in \Binary^n} \mu_x \ketbra{\psi_x}}$
for some probability distribution~${\mu = \{ \mu_x \}_{x \in \Binary^n}}$ over $\Binary^n$
and orthonormal basis~${\{ \ket{\psi_x} \}_{x \in \Binary^n}}$.
Note that
${D(\rho, (I/2)^{\tensor n}) = \SD(\mu, \iota)}$ and ${S(\rho) = H(\mu)}$,
where $\iota$ is the uniform distribution over $\Binary^n$.
Hence, it suffices to show that
the inequality
${H(\mu) \geq (1 - \SD(\mu, \iota)) n - \frac{n}{2^n}}$
holds for any probability distribution~$\mu$.

Let ${\gamma = \SD(\mu, \iota)}$.
By the concavity of the Shannon entropy,
any probability distribution~${\nu = \{ \nu_x \}_{x \in \Binary^n}}$ over $\Binary^n$
that minimizes ${H(\nu)}$ under the condition~${\SD(\nu, \iota) = \gamma}$
can be expressed as follows:
there exist ${x_0, x_1, \ldots, x_k, x_{k+1}}$ in $\Binary^n$ such that
\[
\nu_x =
\begin{cases}
\frac{1}{2^n} + \gamma & \text{if ${x = x_0}$},\\
\frac{1}{2^n} & \text{if ${x \in \{x_1, \ldots, x_k\}}$},\\
\frac{\varepsilon}{2^n} & \text{if ${x = x_{k+1}}$},\\
0 & \text{otherwise},
\end{cases}
\]
where
${k = \floor{2^n(1 - \gamma)} - 1}$
and
${\varepsilon = 2^n(1 - \gamma) - \floor{2^n(1 - \gamma)}}$
(in fact, any probability distribution with statistical distance~$\gamma$ from the uniform distribution~$\iota$
is necessarily a mixture of probability distributions of this type).
It follows that
\[
\begin{split}
H(\nu) 
&
=
\Bigl( \frac{1}{2^n} + \gamma \Bigr) \log \biggl( \frac{1}{\frac{1}{2^n} + \gamma} \biggr)
+
k \frac{n}{2^n}
+
\frac{\varepsilon}{2^n} \log \frac{2^n}{\varepsilon}
\\
&
\geq
\bigl( \floor{2^n(1 - \gamma)} - 1 \bigr) \frac{n}{2^n}
+
\frac{\varepsilon}{2^n} \Bigl( n + \log \frac{1}{\varepsilon} \Bigr)
\\
&
=
(1 - \gamma) n - \frac{n}{2^n} + \frac{\varepsilon}{2^n} \log \frac{1}{\varepsilon}
\\
&
\geq
(1 - \gamma) n - \frac{n}{2^n},
\end{split}
\]
and thus,
the inequality
${H(\mu) \geq (1 - \SD(\mu, \iota)) n - \frac{n}{2^n}}$ holds.

Now we show the second inequality.
Similarly to the first inequality case, 
it suffices to show that
the inequality
${H(\mu) \leq n - \log \frac{1}{1 - \SD(\mu, \iota)} + 2}$
holds for any probability distribution~$\mu$.

Again let ${\gamma = \SD(\mu, \iota)}$.
From the Vajda inequality~\cite{Vaj70IEEEIT}
(see Theorem~4.8 of Ref.~\cite{Dru12ECCC} also),
it holds that
\[
\relent{\mu}{\iota}
\geq 
\frac{1}{\ln 2}
\Bigl( \ln \frac{1}{1 - \gamma} - 1 \Bigr),
\]
where $\relent{\cdot}{\cdot}$ denotes the relative entropy between two probability distributions. 
Since
${\relent{\mu}{\iota} = n - H(\mu)}$,
it follows that
\[
H(\mu)
\leq
n - \frac{1}{\ln 2} \Bigl( \ln \frac{1}{1 - \gamma} - 1 \Bigr)
=
n - \log \frac{1}{1 - \gamma} + \frac{1}{\ln 2}
\leq
n - \log \frac{1}{1 - \gamma} + 2,
\]
as desired.
\end{proof}

% ---------------------------------------------------------------------------
%   Polynomial-Time Uniformly Generated Families of Quantum Circuits
% ---------------------------------------------------------------------------

\subsection{Polynomial-Time Uniformly Generated Families of Quantum Circuits}
\label{Subsection: Uniform QC}

Following conventions,
this paper defines quantum Arthur-Merlin proof systems
in terms of quantum circuits.
In particular, this paper uses the following notion of
polynomial-time uniformly generated families of quantum circuits.

A family~${\{ Q_x \}}$ of quantum circuits is
\emph{polynomial-time uniformly generated}
if there exists a deterministic procedure
that, on every input~$x$, outputs a description of $Q_x$
and runs in time polynomial in $\abs{x}$.
It is assumed that the circuits in such a family are composed of gates
in some reasonable, universal, finite set of quantum gates.
Furthermore, it is assumed that the number of gates in any circuit
is not more than the length of the description of that circuit.
Therefore $Q_x$ must have size polynomial in $\abs{x}$.
For convenience,
we may identify a circuit~$Q_x$ with the unitary operator it induces.

For the results in which perfect completeness is concerned,
this paper assumes a gate set
with which the Hadamard and any classical reversible transformations
can be exactly implemented.
Note that this assumption is satisfied by many standard gate sets
such as the Shor basis~\cite{Sho96FOCS}
consisting of the Hadamard, controlled-$i$-phase-shift, and Toffoli gates,
and the gate set consisting of the Hadamard, Toffoli, and NOT gates~\cite{Shi02QIC, Aha03arXiv}.
Moreover, as the Hadamard transformation in some sense can be viewed as
a quantum analogue of the classical operation of flipping a fair coin,
our assumption would be the most natural quantum correspondence 
to the tacit classical assumption in randomized complexity theory
that fair coins and perfect logical gates are available.
Hence we believe that our condition is very reasonable and not restrictive.

Since non-unitary and unitary quantum circuits
are equivalent in computational power~\cite{AhaKitNis98STOC},
it is sufficient to treat only unitary quantum circuits,
which justifies the above definition.
Nevertheless, for readability,
most procedures in this paper will be described
using intermediate projective measurements
and unitary operations conditioned on the outcome of the measurements.
All of these intermediate measurements can be deferred
to the end of the procedure by a standard technique
so that the procedure becomes implementable with a unitary circuit.

% ---------------------------------------------------------------------------
%   Generalized Quantum Arthur-Merlin Proof Systems
% ---------------------------------------------------------------------------

\subsection{Generalized Quantum Arthur-Merlin Proof Systems}
\label{Subsection: Generalized QAM}

A generalized quantum Arthur-Merlin (QAM) proof system consists of
a polynomial-time quantum verifier and an all-powerful quantum prover.

For any constant~${m \geq 1}$ and any $t_j$ in ${\{\rmc, \rmq\}}$ 
for each $j$ in $\{1, \ldots, m\}$,
a generalized QAM proof system is of ${t_m \cdots t_1}$-QAM type
if the message at the ${(m-j+1)}$st turn is quantum (resp. is restricted to classical)
for each $j$ such that ${t_j = \rmq}$ (resp. ${t_j = \rmc}$).

Formally, an $m$-turn quantum verifier~$V$ for generalized quantum Arthur-Merlin proof systems
is a polynomial-time computable mapping of the form~$\function{V}{\Binary^\ast}{\Binary^\ast}$.
For each $x$ in $\Binary^\ast$,
${V(x)}$ is interpreted as a description of a quantum circuit
acting on ${(q_\sfV(\abs{x}) + m q_\sfM(\abs{x}))}$ qubits
with a specification of a ${q_\sfV(\abs{x})}$-qubit quantum register~$\sfV$
and a ${q_\sfM(\abs{x})}$-qubit quantum register~$\sfM_j$ for each $j$ in ${\{1, \ldots, m\}}$,
for some polynomially bounded functions~$\function{q_\sfV, q_\sfM}{\Nonnegative}{\Natural}$.
One of the qubits in $\sfV$ is designated as an output qubit.
At the ${(m-j+1)}$st turn for any even~$j$ such that ${2 \leq j \leq m-1}$,
$V$ receives a message from a prover, either classical or quantum,
which is stored in the quantum register~$\sfM_{m-j}$. 
If the system is of ${t_m \cdots t_1}$-QAM type,
at the ${(m-j+1)}$st turn for any even~$j$ such that ${2 \leq j \leq m}$,
if ${t_j = \rmc}$, $V$ flips a fair coin ${q_\sfM(\abs{x})}$ times
to obtain a binary string~$r$ of length~${q_\sfM(\abs{x})}$,
then sends $r$ to a prover,
and stores $r$ in the quantum register~$\sfM_{m-j+1}$,
while if ${t_j = \rmq}$, $V$ generates ${q_\sfM(\abs{x})}$~EPR pairs~${\ket{\Phi^+}^{\tensor q_\sfM(\abs{x})}}$,
then sends the second halves of them to a prover,
and stores the first halves of them in ~$\sfM_{m-j+1}$.
Upon receiving a message at the $m$th turn from a prover,
either classical or quantum,
which is stored in the quantum register~$\sfM_m$,
$V$ prepares the ${q_\sfV(\abs{x})}$-qubit quantum register~$\sfV$,
all the qubits of which are initialized to the $\ket{0}$~state.
$V$ then performs the final verification procedure
by applying the circuit~${V(x)}$ to ${(\sfV, \sfM_1, \ldots, \sfM_m)}$
and then measuring the designated output qubit in the computational basis,
where the outcome~$\ket{1}$ is interpreted as ``accept'',
and the outcome~$\ket{0}$ is interpreted as ``reject''.

Similarly, an $m$-turn quantum prover~$P$ for generalized quantum Arthur-Merlin proof systems
is a mapping from $\Binary^\ast$ to a sequence of $\ceil{m/2}$~unitary transformations
with a specification of quantum registers they acts on.
No restrictions are placed on the complexity of $P$.
For each $x$ in $\Binary^\ast$, ${P(x)}$ is interpreted as a sequence of
$\ceil{m/2}$~unitary transformations~${P(x)_{2\ceil{m/2}-1}, \ldots, P(x)_3, P(x)_1}$
acting on ${(q_\sfM(\abs{x}) + q_\sfP(\abs{x}))}$ qubits
with a specification of a ${q_\sfP(\abs{x})}$-qubit quantum register~$\sfP$,
for some polynomially bounded function~$\function{q_\sfM}{\Nonnegative}{\Natural}$
and some function~$\function{q_\sfP}{\Nonnegative}{\Natural}$.
At the beginning of the protocol,
$P$ prepares the ${q_\sfP(\abs{x})}$-qubit quantum register~$\sfP$
(and a ${q_\sfM(\abs{x})}$-qubit quantum register~$\sfM_1$ also, if $m$ is odd).
Without loss of generality, one can assume that
all the qubits in $\sfP$ (and in $\sfM_1$ when $P$ prepares it)
are initialized to the $\ket{0}$~state at the beginning of the protocol.
At the ${(m-j+1)}$st turn for any odd~$j$ such that ${1 \leq j \leq m-1}$,
$P$ receives a message from the verifier, either classical or quantum,
which is stored in the quantum register~$\sfM_{m-j+1}$.
If a system is of ${t_m \cdots t_1}$-QAM type,
at the ${(m-j+1)}$st turn for any odd~$j$ such that ${1 \leq j \leq m}$,
$P$ applies ${P(x)_j}$ to ${(\sfM_{m-j+1}, \sfP)}$.
If ${t_j = \rmc}$, $P$ further measures each qubit in $\sfM_{m-j+1}$ 
in the computational basis.
$P$ then sends $\sfM_{m-j+1}$ to a verifier.

The complexity class~${\GQAM{t_m \cdots t_1}(m,c,s)}$
derived from generalized quantum Arthur-Merlin proof systems of ${t_m \cdots t_1}$-QAM type,
with completeness~$c$ and soundness~$s$,
is defined as follows.

\begin{definition}
Given a constant~${m \in \Natural}$,
functions~$\function{c,s}{\Nonnegative}{[0,1]}$ satisfying ${c > s}$,
and ${t_j \in \{\rmc, \rmq\}}$ for each ${j \in \{1, \ldots, m\}}$,
a promise problem~${A = (A_\yes, A_\no)}$ is in ${\GQAM{t_m \cdots t_1}(m, c, s)}$
if there exists an $m$-turn quantum verifier~$V$ for ${t_m \cdots t_1}$-QAM type systems,
such that, for every input~${x \in \Binary^\ast}$,
\begin{description}
\item[(Completeness)]
  if ${x \in A_\yes}$, then there exists an $m$-turn quantum prover~$P$ for ${t_m \cdots t_1}$-QAM type systems
  that makes $V$ accept $x$ with probability at least ${c(\abs{x})}$,
  and
\item[(Soundness)]
  if ${x \in A_\no}$, then for any $m$-turn quantum prover~$P'$ for ${t_m \cdots t_1}$-QAM type systems,
  $V$ accepts $x$ with probability at most ${s(\abs{x})}$.
\end{description}
\end{definition}

Using this definition,
the classes~${\GQAM{t_m \cdots t_1}(m)}$~and~${\GQAM{t_m \cdots t_1}_1(m)}$
of problems having 
a two-sided bounded error generalized quantum Arthur-Merlin proof system of ${t_m \cdots t_1}$-QAM type,
and that of one-sided bounded error of perfect completeness, respectively,
are defined as follows.

\begin{definition}
Given a constant~${m \in \Natural}$
and ${t_j \in \{\rmc, \rmq\}}$ for each ${j \in \{1, \ldots, m\}}$,
a promise problem~${A = (A_\yes, A_\no)}$ is in ${\GQAM{t_m \cdots t_1}(m)}$
iff $A$ is in ${\GQAM{t_m \cdots t_1}(m, 1- \varepsilon, \varepsilon)}$
for some negligible function~$\function{\varepsilon}{\Nonnegative}{[0,1]}$.
\label{Definition: t_m ... t_1-QAM(m)}
\end{definition}

\begin{definition}
Given a constant~${m \in \Natural}$
and ${t_j \in \{\rmc, \rmq\}}$ for each ${j \in \{1, \ldots, m\}}$,
a promise problem~${A = (A_\yes, A_\no)}$ is in ${\GQAM{t_m \cdots t_1}_1(m)}$
iff $A$ is in ${\GQAM{t_m \cdots t_1}(m, 1, \varepsilon)}$
for some negligible function~$\function{\varepsilon}{\Nonnegative}{[0,1]}$.
\label{Definition: t_m ... t_1-QAM_1(m)}
\end{definition}

In the case where the number of turns is clear,
the parameter~$m$ may be omitted,
e.g., ${\mathrm{cc}\qqQAM(4)}$ may be abbreviated as ${\mathrm{cc}\qqQAM}$.
The following lemmas ensure that Definitions~\ref{Definition: t_m ... t_1-QAM(m)}~and~\ref{Definition: t_m ... t_1-QAM_1(m)}
give a robust definition in terms of completeness and soundness parameters.

\begin{lemma}
For any constant~${m \in \Natural}$, any ${t_1, \ldots, t_m \in \{\rmc, \rmq\}}$,
any polynomially bounded function~$p$,
and any functions~$\function{c, s}{\Nonnegative}{[0,1]}$
satisfying ${c - s \geq \frac{1}{q}}$ for some polynomially bounded function~$q$,
\[
\GQAM{t_m \cdots t_1}(m, c, s)
\subseteq
\GQAM{t_m \cdots t_1}(m, 1 - 2^{-p}, 2^{-p}).
\]
\label{Lemma: amplification}
\end{lemma}

\begin{lemma}
For any constant~${m \in \Natural}$, any ${t_1, \ldots, t_m \in \{\rmc, \rmq\}}$,
any polynomially bounded function~$p$,
and any function~$\function{s}{\Nonnegative}{[0,1]}$
satisfying ${1 - s \geq \frac{1}{q}}$ for some polynomially bounded function~$q$,
\[
\GQAM{t_m \cdots t_1}(m, 1, s)
\subseteq
\GQAM{t_m \cdots t_1}(m, 1, 2^{-p}).
\]
\label{Lemma: amplification for perfect completeness case}
\end{lemma}

The proof of Lemma~\ref{Lemma: amplification} uses the following lemma
(the claim was proved in this form in Ref.~\cite{KobMatYam09CJTCS},
but similar statements are also found in Refs.~\cite{AarBeiDruFefSho09ToC, JaiUpaWat09FOCS}).

\begin{lemma}
Let $\function{c, s}{\Nonnegative}{[0,1]}$ be any functions
that satisfy ${c - s \geq \frac{1}{q}}$ for some polynomially bounded function~$q$,
and let $\Pi$ be any proof system with completeness~$c$ and soundness~$s$.
Fix any polynomially bounded function~$q'$,
and consider another proof system~$\Pi'$ such that,
for every input of length~$n$,
$\Pi'$ carries out ${N = 2 q'(n) (q(n))^2}$ attempts of $\Pi$ in parallel,
and accepts if and only if at least ${\frac{c(n)+s(n)}{2}}$-fraction of these $N$ attempts
results in acceptance in $\Pi$.
Then $\Pi'$ has completeness~${1 - 2^{-q'}}$
and soundness~${\frac{2s}{c+s} \leq 1 - \frac{c-s}{2} \leq 1 - \frac{1}{2q}}$.
\label{Lemma: exp-small completeness error}
\end{lemma}

Now the amplification result for generalized quantum Arthur-Merlin proof systems follows from 
Lemma~\ref{Lemma: exp-small completeness error}
and the perfect parallel repetition theorem for general quantum interactive proof systems~\cite{Gut09PhD}.

\begin{proof}[Proof of Lemma~\ref{Lemma: amplification}]
First, the inclusion~%
${
  \GQAM{t_m \cdots t_1}(m, c, s)
  \subseteq
  \GQAM{t_m \cdots t_1} \bigl( m, 1 - \frac{2^{-(p+1)}}{\ceil{p/(c-s)}}, 1 - \frac{c-s}{2} \bigr)
}$
follows from Lemma~\ref{Lemma: exp-small completeness error}
by taking $q'$ in the statement of Lemma~\ref{Lemma: exp-small completeness error}
as $q'={p + \ceil{\log_2\big(\ceil{\frac{p}{c-s}}\big)} + 1}$.

We show the inclusion~%
${
  \GQAM{t_m \cdots t_1} \bigl( m, 1 - \frac{2^{-(p+1)}}{\ceil{p/(c-s)}}, 1 - \frac{c-s}{2} \bigr)
  \subseteq
  \GQAM{t_m \cdots t_1}(m, 1 - 2^{-p}, 2^{-p})
}$
to complete the proof.

Fix any protocol~$\Pi$ of ${t_m \cdots t_1\textrm{-QAM}(m)}$ proof systems,
and consider the $k$-fold repetition~$\Pi^{\tensor k}$ of $\Pi$,
where Arthur runs $k$~attempts of $\Pi$ in parallel,
and accepts if and only if all of the $k$~attempts result in acceptance in the original~$\Pi$.
We claim that the maximum acceptance probability in $\Pi^{\tensor k}$ is exactly $a^k$
if the maximum acceptance probability in $\Pi$ was $a$.
To show this claim, consider another protocol~${Q(\Pi)}$
of $m$-turn (general) quantum interactive proof systems
that exactly simulates $\Pi$ as follows:
the verifier in ${Q(\Pi)}$ behaves exactly the same manner as Arthur in $\Pi$
except that, upon receiving the $j$th message from a prover
(resp. sending the $j$th message to a prover), if ${t_j = \rmc}$ in $\Pi$,
the verifier of ${Q(\Pi)}$ first makes sure that the received message 
(resp. the sent message) 
is indeed classical
by taking a copy of the message by CNOT operations
(and the copied message will never be touched in the rest of the protocol).
This clearly makes it useless for a malicious prover to send a quantum message,
deviating the original protocol~$\Pi$,
and thus, the maximum acceptance probability in ${Q(\Pi)}$ obviously remains $a$.
Now from the perfect parallel repetition theorem for general quantum interactive proofs~\cite{Gut09PhD},
the $k$-fold parallel repetition~${(Q(\Pi))^{\tensor k}}$ of ${Q(\Pi)}$ has
its maximum acceptance probability exactly~$a^k$.
As the protocol~${(Q(\Pi))^{\tensor k}}$ is identical to
the protocol~${Q(\Pi^{\tensor k})}$ of the $m$-turn (general) quantum interactive proof system
that exactly simulates $\Pi^{\tensor k}$,
the maximum acceptance probability in $\Pi^{\tensor k}$ is also $a^k$.
Hence,
letting ${k = 2\ceil{\frac{p}{c-s}}}$,
the desired inclusion~%
${
  \GQAM{t_m \cdots t_1} \bigl( m, 1 - \frac{2^{-(p+1)}}{\ceil{p/(c-s)}}, 1 - \frac{c-s}{2} \bigr)
  \subseteq
  \GQAM{t_m \cdots t_1}(m, 1 - 2^{-p}, 2^{-p})
}$
follows from the $k$-fold repetition.
\end{proof}

Lemma~\ref{Lemma: amplification for perfect completeness case} is proved
in essentially the same manner as in Lemma~\ref{Lemma: amplification}
(Lemma~\ref{Lemma: exp-small completeness error} is not necessary in this case,
which makes the proof slightly simpler).

% ---------------------------------------------------------------------------
%   The Completeness Result
% ---------------------------------------------------------------------------

\section{$\boldsymbol{\qqQAM}$-Completeness of \CITM} 
\label{Section: qq-QAM-completeness of CITM}

This section proves Theorem \ref{Theorem: qq-QAM-completeness of CITM}, which states 
that the \problemfont{Close Image to Totally Mixed} ($\CITM$) problem
is complete for the class~$\qqQAM$.

First, it is proved that ${\CITM(a,b)}$ is in $\qqQAM$
for appropriately chosen parameters~$a$~and~$b$.
The proof is a special case of the proof of
the \problemfont{Close Image} problem being in ${\QIP(2)}$~\cite{Wat02QIP, HayMilWil12arXiv}.

\begin{lemma}
${\CITM(a,b)}$ is in $\qqQAM$
for any constants~${a, b \in [0,1]}$ satisfying ${(1 - a)^2 > 1 - b^2}$.
\label{Lemma: CITM is in qq-QAM}
\end{lemma}

\begin{proof}
Let $Q_x$ be a quantum circuit of an instance~$x$ of ${\CITM(a,b)}$ acting on $q_\all$~qubits
with $q_\inp$~specified input qubits and $q_\out$~specified output qubits.
Without loss of generality, one can assume that the first $q_\inp$~qubits correspond to the input qubits,
and the last $q_\out$~qubits correspond to the output qubits.
Let $U_{Q_x}$ denote the unitary operator induced by $Q_x$.
We construct a verifier~$V$ of the $\mathrm{qq}$-QAM proof system with completeness ${(1-a)^2}$ and soundness~${1 - b^2}$ as follows
(recall that $a$ and $b$ are constants in the interval~${[0,1]}$
such that ${(1-a)^2 > 1 - b^2}$, and thus this $\mathrm{qq}$-QAM proof system is sufficient
for the claim).

Let $\sfS_1$ and $\sfS_2$ be quantum registers of $q_\out$~qubits.
The verifier~$V$ first generates $q_\out$~EPR pairs~$\ket{\Phi^+}^{\tensor q_\out}$ in ${(\sfS_1, \sfS_2)}$
so that the $j$th qubit of $\sfS_1$ and that of $\sfS_2$ form an EPR pair,
for every ${j \in \{1, \ldots, q_\out\}}$. 
Then $V$ sends $\sfS_2$ to the prover.
Upon receiving a quantum register~$\sfR$ of ${(q_\all - q_\out)}$~qubits,
$V$ applies the unitary transformation~$\conjugate{U_{Q_x}}$ to ${(\sfR, \sfS_1)}$.
Letting $\sfA$ be the quantum register consisting of the last ${(q_\all - q_\inp)}$~qubits of the register~${(\sfR, \sfS_1)}$
(i.e., corresponding to the \emph{non-input} qubits of $Q_x$),
$V$ accepts $x$ if and only if all the qubits in $\sfA$ are in the $\ket{0}$~state.
Figure~\ref{Figure: verifier's qq-QAM protocol for CITM} summarizes the protocol
of the verifier~$V$.

\begin{figure}[t!]
\begin{protocol*}{Verifier's $\boldsymbol{\mathrm{qq}}$-QAM Protocol for $\boldsymbol{\CITM(a,b)}$}
\begin{step}
\item
  Prepare $q_\out$~qubit registers~$\sfS_1$ and $\sfS_2$,
  and generate $q_\out$~EPR pairs~$\ket{\Phi^+}^{\tensor q_\out}$ in ${(\sfS_1, \sfS_2)}$
  so that the $j$th qubit of $\sfS_1$ and that of $\sfS_2$ form an EPR pair,
  for every ${j \in \{1, \ldots, q_\out\}}$.
  Send $\sfS_2$ to the prover.
\item
  Receive a ${(q_\all - q_\out)}$-qubit quantum register~$\sfR$ from the prover.
  Apply the unitary transformation~$\conjugate{U_{Q_x}}$ to ${(\sfR, \sfS_1)}$.
  Accept if all the qubits in $\sfA$ are in the $\ket{0}$~state, and reject otherwise,
  where $\sfA$ is the quantum register consisting of the last ${(q_\all - q_\inp)}$~qubits
  of ${(\sfS_1, \sfR)}$ (i.e., the non-input qubits of $Q_x$).
\end{step}
\end{protocol*}
\caption{Verifier's $\mathrm{qq}$-QAM protocol for $\CITM$.}
\label{Figure: verifier's qq-QAM protocol for CITM}
\end{figure}

Let $\calW$ denote the Hilbert space corresponding to the $q_\inp$~input qubits of $Q_x$.

For the completeness, suppose that there exists a quantum state~${\rho \in \Density(\calW)}$ such that
${D(Q_x(\rho), (I/2)^{\tensor q_\out}) \leq a}$.
By Lemma~\ref{Lemma: trace distance and fidelity} (the Fuchs-van-de-Graaf inequalities),
it holds that ${F(Q_x(\rho), (I/2)^{\tensor q_\out}) \geq 1 - a}$.
Consider a ${2 q_\inp}$-qubit pure state~$\ket{\phi_\rho}$ that is a purification of $\rho$
such that $\rho$ is the reduced state obtained by tracing out the first $q_\inp$~qubits of $\ket{\phi_\rho}$
(such a purification always exists).
Then, the ${(q_\all + q_\inp)}$-qubit state
\[
  \ket{\psi_\rho}
  =
  (I^{\tensor q_\inp} \tensor U_{Q_x}) (\ket{\phi_\rho} \tensor \ket{0}^{\tensor (q_\all - q_\inp)})
\]
is necessarily a purification of ${Q_x(\rho)}$,
and thus, the ${(q_\all + q_\inp + q_\out)}$-qubit state~${\ket{\psi'_\rho} = \ket{0}^{\tensor q_\out} \tensor \ket{\psi_\rho}}$
is also a purification of ${Q_x(\rho)}$.
On the other hand, an obvious purification of the $q_\out$-qubit totally mixed state~${(I/2)^{\tensor q_\out}}$
is the ${2 q_\out}$-qubit state~$\ket{\xi}$
that is obtained by rearranging the qubits of $\ket{\Phi^+}^{\tensor q_\out}$
so that the $j$th qubit and the ${(q_\out + j)}$th qubit form an EPR pair
 for every ${j \in \{1, \ldots, q_\out\}}$.
Hence, the ${(q_\all + q_\inp + q_\out)}$-qubit state~${\ket{\xi'} = \ket{0}^{\tensor (q_\all + q_\inp - q_\out)} \tensor \ket{\xi}}$
is also a purification of ${(I/2)^{\tensor q_\out}}$.
As the reduced state consisting of the last $q_\out$~qubits of $\ket{\psi'_\rho}$ is exactly ${Q_x(\rho)}$,
while the reduced state consisting of the last $q_\out$~qubits of $\ket{\xi'}$ is exactly ${(I/2)^{\tensor q_\out}}$,
it follows from Lemma~\ref{Lemma: Uhlmann's theorem} (Uhlmann's theorem) that 
\[
F \bigl( Q_x(\rho), (I/2)^{\tensor q_\out} \bigr)
=
\max_U \bigabs{\bra{\psi'_\rho} (U \tensor I^{\tensor q_\out}) \ket{\xi'}}
\]
where the maximum is taken over all unitary operators~$U$ acting on ${(q_\all + q_\inp)}$~qubits.
This in particular implies that there exists a unitary operator~$U_P$ acting on ${(q_\all + q_\inp)}$~qubits such that
\[
\bigabs{\bra{\psi'_\rho} (U_P \tensor I^{\tensor q_\out}) \ket{\xi'}}
=
F \bigl( Q_x(\rho), (I/2)^{\tensor q_\out} \bigr)
\geq
1 - a.
\]
Thus, if a prover prepares $\ket{0}^{\tensor (q_\all + q_\inp - q_\out)}$
in his/her private quantum register~$\sfP$ of ${(q_\all + q_\inp - q_\out)}$~qubits,
applies $U_P$ to ${(\sfP, \sfS_2)}$ after having received $\sfS_2$,
and sends the last ${(q_\all - q_\out)}$~qubits of ${(\sfP, \sfS_2)}$ back to the verifier,
the probability of acceptance is
\[
\begin{split}
\hspace{5mm}&\hspace{-5mm}
\bignorm{
  \bigl( I^{\tensor (2q_\inp + q_\out)} \tensor (\ketbra{0})^{\tensor (q_\all - q_\inp)} \bigr)
  (I^{\tensor (q_\inp + q_\out)} \tensor \conjugate{U_{Q_x}})
  (U_P \tensor I^{\tensor q_\out})
  \ket{\xi'}
}^2
\\
&
\geq
\bignorm{
  \bigl( (\ketbra{0})^{\tensor q_\out} \tensor \ketbra{\phi_\rho} \tensor (\ketbra{0})^{\tensor (q_\all - q_\inp)} \bigr)
  (I^{\tensor (q_\inp + q_\out)} \tensor \conjugate{U_{Q_x}})
  (U_P \tensor I^{\tensor q_\out})
  \ket{\xi'}
}^2
\\
&
=
\bigabs{\bra{\psi'_\rho} (U_P \tensor I^{\tensor q_\out}) \ket{\xi'}}^2
\\
&
\geq
(1 - a)^2,
\end{split}
\]
where the first inequality follows from the fact that
${(\ketbra{0})^{\tensor q_\out} \tensor \ketbra{\phi_\rho} \tensor I^{\tensor (q_\all - q_\inp)}}$
is a projection operator.
This implies the completeness~${(1 - a)^2}$ of the constructed proof system.

For the soundness, suppose that for any quantum state~${\rho \in \Density(\calW)}$,
it holds that ${D(Q_x(\rho), (I/2)^{\tensor q_\out}) \geq b}$.
Let $P'$ be any prover who uses his/her private quantum register~$\sfP'$ of $q$~qubits,
for arbitrarily large integer~$q$.
Without loss of generality, one can assume that all the qubits in $\sfP'$ are
in the $\ket{0}$~state at the beginning of the protocol. 
Let $U_{P'}$ be the unitary operator acting on ${(q + q_\out)}$~qubits
which $P'$ applies to ${(\sfP', \sfS_2)}$ after having received $\sfS_2$,
and let $\ket{\phi}$ be the ${(q + 2 q_\out)}$-qubit state defined by
\[
\ket{\phi}
=
(I^{\tensor (q - q_\all + 2 q_\out)} \tensor \conjugate{U_{Q_x}})
(U_{P'} \tensor I^{\tensor q_\out})
\ket{\xi''},
\]
where $\ket{\xi''}$ is the ${(q + 2 q_\out)}$-qubit state defined as
${\ket{\xi''} = \ket{0}^{\tensor q} \tensor \ket{\xi}}$.
Define the projection operator~$\Pi_\acc$ by
${
\Pi_\acc
=
I^{\tensor (q - q_\all + q_\inp + 2 q_\out)} \tensor (\ketbra{0})^{\tensor (q_\all - q_\inp)}
}$.
Then, the ${(q + 2 q_\out)}$-qubit state~$\ket{\psi}$ defined by
${
\ket{\psi} = \frac{1}{\norm{\Pi_\acc \ket{\phi}}} \Pi_\acc \ket{\phi}
}$
must be written as
${
\ket{\psi} = \ket{\psi'} \tensor \ket{0}^{\tensor (q_\all - q_\inp)}
}$
for some ${(q - q_\all + q_\inp + 2 q_\out)}$-qubit state~$\ket{\psi'}$,
as ${\Pi_\acc \ket{\psi} = \ket{\psi}}$ holds.

As ${D(Q_x(\rho), (I/2)^{\tensor q_\out}) \geq b}$ for any quantum state~${\rho \in \Density(\calW)}$,
from Lemma~\ref{Lemma: trace distance and fidelity} (the Fuchs-van-de-Graaf inequalities),
it holds that
${F(Q_x(\rho), (I/2)^{\tensor q_\out}) \leq \sqrt{1 - b^2}}$ for any quantum state~${\rho \in \Density(\calW)}$.
This in particular implies that
\begin{equation}
\begin{split}
\abs{\braket{\psi}{\phi}}
&
=
\bigabs{
  (\bra{\psi'} \tensor \bra{0}^{\tensor (q_\all - q_\inp)})
  (I^{\tensor (q - q_\all + 2 q_\out)} \tensor \conjugate{U_{Q_x}})
  (U_{P'} \tensor I^{\tensor q_\out})
  \ket{\xi''}
}
\\
&
\leq
F(Q_x(\rho_{\psi'}), (I/2)^{\tensor q_\out})
\leq
\sqrt{1 - b^2},
\label{Equation: fidelity bound}
\end{split}
\end{equation}
where ${\rho_{\psi'} \in \Density(\calW)}$ is the reduced state of $\ket{\psi'}$
obtained by tracing out
all but the last $q_\inp$~qubits, and we have used the fact that
the reduced state consisting of the last $q_\out$~qubits of $\ket{\xi''}$ is exactly ${(I/2)^{\tensor q_\out}}$
on which $U_{P'}$ never acts.
As the acceptance probability~$p_{P'}$ with this prover~$P'$ is exactly
${
\norm{\Pi_\acc \ket{\phi}}^2
}$,
while
${
\norm{\Pi_\acc \ket{\phi}}
=
\frac{1}{\norm{\Pi_\acc \ket{\phi}}} \abs{\bra{\phi} \Pi_\acc \ket{\phi}}
=
\abs{\braket{\psi}{\phi}}
}$,
it holds from Eq.~(\ref{Equation: fidelity bound}) that ${p_{P'} \leq 1 - b^2}$,
and the soundness follows.
\end{proof}

Now the $\CITM$ problem is proved to be hard for $\qqQAM$. 

\begin{lemma}
For any constants~$a$~and~$b$ such that ${0 < a < b < 1}$, 
${\CITM(a,b)}$ is hard for $\qqQAM$ under polynomial-time many-one reduction. 
\label{Lemma: qq-QAM-hardness of CITM}
\end{lemma}

\begin{proof}
Let ${A = (A_\yes, A_\no)}$ be a problem in $\qqQAM$.
Then $A$ has a $\mathrm{qq}$-QAM proof system with completeness~$c$ and soundness~$s$
for some constants~$c$~and~$s$ chosen later satisfying ${0 < s < c < 1}$.
Let $V$ be the quantum verifier witnessing this proof system.
Fix an input~$x$, and 
let $\sfV$ and $\sfM$ be quantum registers consisting of $q_\sfV$~and~$q_\sfM$~qubits, respectively,
where $\sfV$ corresponds to the private qubits of $V$
and $\sfM$ corresponds to the message qubits $V$ would receive on input~$x$.
Without loss of generality, one can assume that
the first qubit of $\sfV$ is the output qubit of $V$,
and the last $q_\sfS$~qubits of $\sfV$ form the quantum register~$\sfS$
corresponding to the halves of the EPR pairs $V$ would keep until the final verification procedure is performed.
Let $\overline{\sfS}$ be the quantum register of ${(q_\sfV - q_\sfS)}$~qubits
consisting of the first ${(q_\sfV - q_\sfS)}$~qubits of $\sfV$ 
(i.e., all the private qubits of $V$ but those belonging to $\sfS$).
Denote by $V_x$ the unitary operator induced by this $V$ on input~$x$.

We construct a quantum circuit~$Q_x$ that exactly implements the following algorithm.
The circuit~$Q_x$ expects to receive a ${(q_\sfS + q_\sfM)}$-qubit state as its input,
and prepares the quantum registers~${\sfV = (\overline{\sfS}, \sfS)}$~and~$\sfM$,
where the input state is expected to be stored in ${(\sfS, \sfM)}$.
Then with probability one-half,
$Q_x$ just outputs the state in the register~$\sfS$.
Otherwise $Q_x$ performs $V_x$ over ${(\sfV, \sfM) = (\overline{\sfS}, \sfS, \sfM)}$,
and outputs the totally mixed state~${(I/2)^{\tensor q_\sfS}}$
if the first qubit of $\sfV$ is in state~$\ket{1}$
(i.e., if the system is in an accepting state of the original verifier~$V$),
and outputs $(\ketbra{0})^{\tensor q_\sfS}$
if the first qubit of $\sfV$ is in state~$\ket{0}$
(i.e., if the system is in a rejecting state of the original verifier~$V$).
Figure~\ref{Figure: construction of Q_x (modified)} summarizes the construction of the circuit~$Q_x$.

\begin{figure}[t!]
\begin{algorithm*}{Algorithm Corresponding to Quantum Circuit $\boldsymbol{Q_x}$}
\begin{step}
\item
  Prepare the quantum registers~$\sfV$~and~$\sfM$, each of $q_\sfV$~and~$q_\sfM$~qubits, respectively.
  Denote by $\sfS$ and $\overline{\sfS}$ the quantum registers consisting of
  the last $q_\sfS$ and first ${(q_\sfV - q_\sfS)}$~qubits of $\sfV$, respectively.
  The last ${(q_\sfS + q_\sfM)}$~qubits in ${(\sfV, \sfM) = (\overline{\sfS}, \sfS, \sfM)}$
  (i.e., all the qubits in ${(\sfS, \sfM)}$)
  are designated as the input qubits,
  while the last $q_\sfS$~qubits of ${\sfV = (\overline{\sfS}, \sfS)}$
  (i.e., all the qubits in $\sfS$)
  are designated as the output qubits.
\item
  Flip a fair coin, and proceed to Step~2.1 if it results in ``Heads'',
  and proceed to Step~2.2 if it results in ``Tails''.
  \begin{step}
  \item
    Output all the qubits in $\sfS$.
  \item
    Perform $V_x$ over ${(\sfV, \sfM) = (\overline{\sfS}, \sfS, \sfM)}$.
    If the first qubit of $\sfV$ is in state~$\ket{1}$,
    output the totally mixed state~${(I/2)^{\tensor q_\sfS}}$
    (by first generating the totally mixed state using fresh ancillae,
    and then swapping the qubits in $\sfS$ with the generated totally mixed state),
    and output $\ket{0}^{\tensor q_\sfS}$ otherwise
    (by swapping the qubits in $\sfS$ with $q_\sfS$~fresh ancillae).
  \end{step}
\end{step}
\end{algorithm*}
\caption{The construction of the quantum circuit~$Q_x$.}
\label{Figure: construction of Q_x (modified)}
\end{figure}

First suppose that $x$ is in $A_\yes$.
Then there exists a quantum prover~$P$ who makes $V$ accept with probability at least $c$.
Let $\rho_x$ be the ${(q_\sfS + q_\sfM)}$-qubit state in ${(\sfS, \sfM)}$
just after $V$ has received a response from $P$ on input~$x$.
Note that the reduced state in $\sfS$ of $\rho_x$ when tracing out all the qubits in $\sfM$
is exactly ${(I/2)^{\tensor q_\sfS}}$,
as $P$ has never touched the qubits in ${\sfV = (\overline{\sfS}, \sfS)}$.
Let $\rho'_x$ be the ${(q_\sfV + q_\sfM)}$-qubit state in ${(\sfV, \sfM) = (\overline{\sfS}, \sfS, \sfM)}$
defined by ${\rho'_x = \ketbra{0}^{\tensor (q_\sfV - q_\sfS)} \tensor \rho_x}$,
and let $\Pi_\acc$ be the projection operator defined by
${\Pi_\acc = \ketbra{1} \tensor I^{\tensor (q_\sfV + q_\sfM - 1)}}$.

Then ${p_\acc = \tr \Pi_\acc V_x \rho'_x \conjugate{V_x}}$ is exactly the acceptance probability with this prover~$P$, which is at least $c$,
and $Q_x$ outputs the state
\[
\xi = p_\acc (I/2)^{\tensor q_\sfS} + (1 - p_\acc) (\ketbra{0})^{\tensor q_\sfS}
\]
in Step~2.2, when $\rho_x$ is given as an input to $Q_x$.
On the other hand, $Q_x$ clearly outputs the totally mixed state ${(I/2)^{\tensor q_\sfS}}$ in Step~2.1,
when $\rho_x$ is given as an input to $Q_x$.
Hence, given the input state~$\rho_x$, the circuit~$Q_x$ outputs the state
\[
Q_x(\rho_x)
=
\frac{1}{2} (I/2)^{\tensor q_\sfS} + \frac{1}{2} \xi
=
\frac{1}{2} (1 + p_\acc) (I/2)^{\tensor q_\sfS} + \frac{1}{2} (1 - p_\acc) (\ketbra{0})^{\tensor q_\sfS}.
\]
Therefore,
\[
\bigtrnorm{Q_x(\rho_x) - (I/2)^{\tensor q_\sfS}}
=
\frac{1}{2} (1 - p_\acc)
\bigtrnorm{(\ketbra{0})^{\tensor q_\sfS} - (I/2)^{\tensor q_\sfS}},
\]
which implies that
\[
D \bigl( Q_x(\rho_x), (I/2)^{\tensor q_\sfS} \bigr)
=
\frac{1}{2} (1 - p_\acc)
D \bigl( (\ketbra{0})^{\tensor q_\sfS}, (I/2)^{\tensor q_\sfS} \bigr)
\leq
\frac{1}{2} (1 - p_\acc)
\leq
\frac{1}{2} (1 - c).
\]
Hence, choosing ${c \geq 1 - 2a}$,
the inequality~${D \bigl( Q_x(\rho_x), (I/2)^{\tensor q_\sfS} \bigr) \leq a}$ holds.

Now suppose that $x$ is in $A_\no$.
Then $V$ accepts with probability at most $s$ no matter which quantum prover he communicates with.
Let $\rho$ be any ${(q_\sfS + q_\sfM)}$-qubit state in ${(\sfS, \sfM)}$,
and consider the reduced state~$\rho'$ in $\sfS$ of $\rho$.
As before, let $\Pi_\acc$ be the projection operator defined by
${\Pi_\acc = \ketbra{1} \tensor I^{\tensor (q_\sfV + q_\sfM - 1)}}$.
The state~${Q_x(\rho)}$ that $Q_x$ outputs when the input state was $\rho$
is given by
\[
Q_x(\rho)
=
\frac{1}{2} \rho'
+
\frac{1}{2} \bigl[ p'_\acc (I/2)^{\tensor q_\sfS} + (1 - p'_\acc) (\ketbra{0})^{\tensor q_\sfS} \bigr],
\]
where
${p'_\acc = \tr \Pi_\acc V_x \bigl( (\ketbra{0})^{\tensor (q_\sfV - q_\sfS)} \tensor \rho \bigr) \conjugate{V_x}}$
is the probability that $Q_x$ outputs the totally mixed state in Step~2.2,
when given the input state~$\rho$.

If ${D(\rho', (I/2)^{\tensor q_\sfS}) \geq 1 - \frac{1}{\sqrt{5}}}$,
by Lemma~\ref{Lemma: mixture of states and trace distance},
the state~${Q_x(\rho)}$ that $Q_x$ outputs when the input state was $\rho$
satisfies that
\[
D \bigl( Q_x(\rho), (I/2)^{\tensor q_\sfS} \bigr)
\geq
D(\rho', (I/2)^{\tensor q_\sfS}) - \frac{1}{2}
\geq
\frac{1}{2} - \frac{1}{\sqrt{5}}.
\]

On the other hand,
if ${D(\rho', (I/2)^{\tensor q_\sfS}) < 1 - \frac{1}{\sqrt{5}}}$,
consider any purification~$\ket{\phi_\rho}$ in ${(\sfS, \sfM, \sfP)}$ of $\rho$,
where $\sfP$ is a quantum register sufficiently large for the purification.
Note that $\ket{\phi_\rho}$ is also a purification of the reduced state~$\rho'$ of $\rho$,
and thus, by Lemma~\ref{Lemma: Uhlmann's theorem} (Uhlmann's theorem),
there should be a purification~$\ket{\phi_\legal}$ in ${(\sfS, \sfM, \sfP)}$ of
the totally mixed state~${(I/2)^{\tensor q_\sfS}}$ such that
\[
F(\ketbra{\phi_\rho}, \ketbra{\phi_\legal})
=
F(\rho', (I/2)^{\tensor q_\sfS}).
\]
Therefore, the reduced state~$\rho_\legal$ in ${(\sfV, \sfM) = (\overline{\sfS}, \sfS, \sfM)}$
of the state~${(\ketbra{0})^{\tensor (q_\sfV - q_\sfS)} \tensor \ketbra{\phi_\legal}}$
must satisfy that
\[
F((\ketbra{0})^{\tensor (q_\sfV - q_\sfS)} \tensor \rho, \rho_\legal)
=
F(\ketbra{\phi_\rho}, \ketbra{\phi_\legal})
=
F(\rho', (I/2)^{\tensor q_\sfS}),
\]
and thus,
Lemma~\ref{Lemma: trace distance and fidelity} (the Fuchs-van-de-Graaf inequalities)
implies that
\begin{equation}
\begin{split}
D((\ketbra{0})^{\tensor (q_\sfV - q_\sfS)} \tensor \rho, \rho_\legal)
&
\leq
\sqrt{1 - F((\ketbra{0})^{\tensor (q_\sfV - q_\sfS)} \tensor \rho, \rho_\legal)^2}
\\
&
=
\sqrt{1 - F(\rho', (I/2)^{\tensor q_\sfS})^2}
\leq
\sqrt{1 - \bigl(1 - D(\rho', (I/2)^{\tensor q_\sfS}) \bigr)^2}
<
\frac{2}{\sqrt{5}}.
\end{split}
\label{Inequality: trace distance to legal}
\end{equation}
As $\rho_\legal$ is a legal state
that can appear in ${(\sfV, \sfM) = (\overline{\sfS}, \sfS, \sfM)}$
of the starting $\mathrm{qq}$-QAM system just before the final verification procedure of $V$,
from the soundness property of the system, it holds that
${
\tr \Pi_\acc V_x \rho_\legal \conjugate{V_x} \leq s
}$.
Thus, from Lemma~\ref{Lemma: trace distance and probability}
together with Eq.~(\ref{Inequality: trace distance to legal}),
the probability~$p'_\acc$ that $Q_x$ outputs the totally mixed state in Step~2.2,
when given the input state~$\rho$,
is bounded from above by
\[
p'_\acc
\leq
s + D((\ketbra{0})^{\tensor (q_\sfV - q_\sfS)} \tensor \rho, \rho_\legal)
<
s + \frac{2}{\sqrt{5}}.
\]
This implies that, when the input state was $\rho$,
the probability~$p'_0$ that $Q_x$ outputs the state~${(\ketbra{0})^{\tensor q_\sfS}}$
is bounded by
\[
p'_0 = \frac{1}{2} (1 - p'_\acc) > \frac{1}{2} \Bigl( 1 - s - \frac{2}{\sqrt{5}} \Bigr),
\]
and thus, by Lemma~\ref{Lemma: mixture of states and trace distance},
the state~${Q_x(\rho)}$ that $Q_x$ outputs when the input state was $\rho$
satisfies that
\[
\begin{split}
D \bigl( Q_x(\rho), (I/2)^{\tensor q_\sfS} \bigr)
&
\geq
D((\ketbra{0})^{\tensor q_\sfS}, (I/2)^{\tensor q_\sfS}) - (1 - p'_0)
\\
&
>
(1 - 2^{- q_\sfS}) - \Bigl[ 1 - \frac{1}{2} \Bigl( 1 - s - \frac{2}{\sqrt{5}} \Bigr) \Bigr]
=
\frac{1}{2} - \frac{1}{\sqrt{5}} - \frac{s}{2} - 2^{- q_\sfS}.
\end{split}
\]

Hence, no matter which state~$\rho$ given as input,
it holds that
\[
D \bigl( Q_x(\rho), (I/2)^{\tensor q_\sfS} \bigr)
>
\min
\Bigl\{
\frac{1}{2} - \frac{1}{\sqrt{5}},
\>\>
\frac{1}{2} - \frac{1}{\sqrt{5}} - \frac{s}{2} - 2^{- q_\sfS}
\Bigr\}
=
\frac{1}{2} - \frac{1}{\sqrt{5}} - \frac{s}{2} - 2^{- q_\sfS}.
\]
Without loss of generality, one can assume that ${q_\sfS \geq 10}$,
and thus,
by choosing ${s \leq 2^{-9}}$,
the inequality~${D \bigl( Q_x(\rho), (I/2)^{\tensor q_\sfS} \bigr) > 1/20}$ holds for any~$\rho$.

This completes the proof of the $\qqQAM$-hardness of ${\CITM(a, 1/20)}$
for any positive constant~${a < 1/20}$.

The $\qqQAM$-hardness of ${\CITM(a, b)}$ for any constants~$a$~and~$b$ satisfying ${0 < a < b < 1}$ follows
by first creating an instance~$Q_x$ of ${\CITM(a/k, 1/20)}$ for some constant~${k \in \Natural}$ according to the construction above, 
and then constructing another circuit~$Q'_x$ that places $k$~copies of $Q_x$ in parallel.
Indeed, Lemma~\ref{Lemma: polarization of minimum output trace distance}
ensures that $Q'_x$ is an instance of ${\CITM(a,b)}$,
by taking ${k = \bigceil{2 \frac{\ln(1/(1-b))}{\ln (400/399)}}}$
and considering the transformation~$\Phi$ induced by $Q_x$
and the transformation~$\Psi$ that receives an input state of ${(q_\sfS + q_\sfM)}$ qubits
and always outputs the totally mixed state~${(I/2)^{\tensor q_\sfS}}$ regardless of the input.
\end{proof}

From Lemmas~\ref{Lemma: CITM is in qq-QAM}~and~\ref{Lemma: qq-QAM-hardness of CITM}, Theorem~\ref{Theorem: qq-QAM-completeness of CITM} follows.

Note that, with essentially the same proofs as those for Lemmas~\ref{Lemma: CITM is in qq-QAM}~and~\ref{Lemma: qq-QAM-hardness of CITM},
one can show that for any $b$ in~${(0,1)}$,
${\CITM(0,b)}$ is in $\qqQAM_1$ and is hard for $\qqQAM_1$,
and thus, the following corollary holds.
\begin{corollary}
For any constant~$b$ in ${(0,1)}$,
${\CITM(0,b)}$ is $\qqQAM_1$-complete under polynomial-time many-one reduction.
\label{Corollary: qq-QAM_1-completeness of CITM(0,b)}
\end{corollary}

\begin{remark}
The proofs of this section actually also show that the variant of the $\CITM$ problem where
the number of output qubits of the circuit is a fixed constant independent of instances
is complete for the class~$\EPRQMA{\const}$ introduced in Ref.~\cite{KobLeGNis13ITCS},
and thus, it is $\QMA$-complete since ${\EPRQMA{\const} = \QMA}$~\cite{BeiShoWat11ToC}.
\end{remark}

% ---------------------------------------------------------------------------
%   Collapse Theorem for qq-QAM
% ---------------------------------------------------------------------------

\section{Collapse Theorem for $\boldsymbol{\qqQAM}$}
\label{Section: c...cqq-QAM(m) = qq-QAM}

This section proves Theorem~\ref{Theorem: c...cqq-QAM(m) = qq-QAM}, the
quantum analogue of Babai's collapse theorem~\cite{Bab85STOC} stating that
${\mathrm{c \cdots c}\qqQAM(m) = \qqQAM}$ for any constant~${m \geq 2}$.

First, it is proved that for any constant~${m \geq 4}$,
${\mathrm{c \cdots c}\qqQAM(m) \subseteq \mathrm{cc}\qqQAM}$ holds,
meaning that the first ${(m-4)}$~classical turns can be removed.
The proof essentially relies on the observation that the techniques used in 
the classical result by Babai~\cite{Bab85STOC} can be applied in the 
quantum setting as well.\footnotemark

\footnotetext{
  In Ref.~\cite{BabMor88JCSS}, the journal version of Ref.~\cite{Bab85STOC},
  a more efficient protocol (the {\em speedup theorem}) is given to reduce the number of turns,
  but it is more complicated, and not necessary for our purpose.
}

\begin{lemma}
For any constant~${m \geq 4}$,
${
\mathrm{c \cdots c}\qqQAM(m) \subseteq \mathrm{cc}\qqQAM
}$.
\label{Lemma: c...cqq-QAM(m) is in ccqq-QAM}
\end{lemma}

\begin{proof}
It suffices to show that for any odd constant~${m \geq 5}$, 
${\mathrm{c \cdots c}\qqQAM(m) \subseteq \mathrm{c \cdots c}\qqQAM(m-1)}$,
and for any even constant~${m \geq 6}$,
${\mathrm{c \cdots c}\qqQAM(m) \subseteq \mathrm{c \cdots c}\qqQAM(m-2)}$.

Let ${A = (A_\yes, A_\no)}$ be a problem in ${\mathrm{c \cdots c}\qqQAM(m)}$.
By Lemma~\ref{Lemma: amplification},
$A$ has an $m$-turn $\mathrm{c \cdots cqq}$-QAM proof system~$\Pi$
with completeness~${1 - 2^{-8}}$ and soundness~$2^{-8}$. 
Without loss of generality, one can assume that,
for every input of length~$n$,
every classical message exchanged consists of ${l(n)}$ bits
for some polynomially bounded function~$l$.

First consider the case with odd~$m$, where the first turn is for the prover.
Fix an input~$x$, and let ${w_x(y,r)}$ be the maximum of the probability that
a prover can make the verifier accept,
under the condition that the first message from the prover is ${y \in \Binary^{l(\abs{x})}}$
and the second message from the verifier is ${r \in \Binary^{l(\abs{x})}}$.
Then, the maximum acceptance probability in the system~$\Pi$ is given by 
${p_x = \max_y \{ \expectation{w_x(y,r)} \}}$,
where the expectation is taken over the uniform distribution 
with respect to ${r \in \Binary^{l(\abs{x})}}$.
Note that ${p_x \geq 1 - 2^{-8}}$ if $x$ is in $A_\yes$,
and ${p_x \leq 2^{-8}}$ if $x$ is in $A_\no$.

Consider the ${(m-1)}$-turn $\mathrm{c \cdots cqq}$-QAM proof system~$\Pi'$
specified by the following protocol of the verifier:
At the first turn,
the verifier sends ${k(\abs{x})}$~strings~${r_1, \ldots, r_{k(\abs{x})}}$ chosen uniformly at random from $\Binary^{l(\abs{x})}$,
for some polynomially bounded function~$k$.
Upon receiving a pair of strings~$y$~and~$z$ in $\Binary^{l(\abs{x})}$
at the third turn,
the verifier enters the simulations of the last ${(m-3)}$~turns of communications of $\Pi$,
by running in parallel ${k(\abs{x})}$~attempts of such simulations, 
where the $j$th attempt assumes that
the first three messages in the original $\Pi$ were $y$, $r_j$, and $z$,
respectively, for each ${j \in \{1, \ldots, k(\abs{x}) \}}$.
The verifier accepts if and only if more than ${k(\abs{x})/2}$ attempts result in acceptance in the original $\Pi$.
Figure~\ref{Figure: reducing the number of turns by one, for odd m}
summarizes the protocol of this verifier in $\Pi'$.

\begin{figure}[t!]
\begin{protocol*}{Verifier's Protocol for Reducing the Number of Turns by One (for Odd $\boldsymbol{m}$)}
\begin{step}
\item
  Send ${k(\abs{x})}$~strings~${r_1, \ldots, r_{k(\abs{x})}}$,
  each chosen uniformly at random from $\Binary^{l(\abs{x})}$,
  to the prover,
  for some polynomially bounded function~$k$.
\item
  Receive a pair of strings~$y$~and~$z$ in $\Binary^{l(\abs{x})}$ from the prover.
  Run in parallel ${k(\abs{x})}$~attempts of the ${(m-3)}$-turn protocol
  that simulates the last ${(m-3)}$~turns of communications of
  the original $m$-turn $\mathrm{c \cdots cqq}$-QAM proof system~$\Pi$ on input~$x$,
  where the $j$th attempt assumes that
  the first three messages in the original $\Pi$ were $y$, $r_j$, and $z$,
  respectively, for each ${j \in \{1, \ldots, k(\abs{x}) \}}$.
  Accept if more than ${k(\abs{x})/2}$ attempts result in acceptance in the simulations of $\Pi$,
  and reject otherwise.
\end{step}
\end{protocol*}
\caption{Verifier's protocol in $\Pi'$ for reducing the number of turns by one when $m$ is odd.}
\label{Figure: reducing the number of turns by one, for odd m}
\end{figure}

In fact, the construction of this proof system~$\Pi'$
is exactly the same as in Ref.~\cite{Bab85STOC}
except that the last two messages exchanged are quantum
and the final verification of the verifier is a polynomial-time quantum computation
in the present case.
The analysis in Ref.~\cite{Bab85STOC} works also in the present case,
since it only relies on the fact that ${w_x(y,r)}$ gives the conditional probability defined above,
and the perfect parallel repetition theorem holds for general quantum interactive proof systems~\cite{Gut09PhD}.
In particular,
the following property holds also in the present case 
(see Lemmas~3.3~and~3.4 of Ref.~\cite{Bab85STOC}). 

\begin{claim}\label{claim1}
The maximum acceptance probability~$p'_x$ in $\Pi'$ satisfies that
\[
1 - 2^{k(\abs{x})}(1 - p_x)^{k(\abs{x})/2} \leq p'_x \leq 2^{k(\abs{x}) + l(\abs{x})} p_x^{k(\abs{x})/2}.
\]
\end{claim}

Now let ${k = \bigceil{\frac{2+l}{3}}}$.
If $x$ is in $A_\yes$, then the maximum acceptance probability~$p'$ in $\Pi'$ is at least
\[
1 - 2^{k(\abs{x})} (1 - p_x)^{k(\abs{x})/2}
\geq
1 - 2^{k(\abs{x})} (2^{-8})^{k(\abs{x})/2}
\geq
1 - \frac{1}{2^{l(\abs{x}) + 2}}
\geq
\frac{3}{4},
\]
while if $x$ is in $A_\no$, then the maximum acceptance probability~$p'$ in $\Pi'$ is at most
\[
2^{k(\abs{x}) + l(\abs{x})} p_x^{k(\abs{x})/2}
\leq
2^{k(\abs{x}) + l(\abs{x})} (2^{-8})^{k(\abs{x})/2}
\leq
\frac{1}{4},
\]
which completes the proof for the case with odd~$m$.

Next consider the case with even~$m$, where the first message is a random string from a verifier.
Let $\Pi^{(-1)}$ be the ${(m-1)}$-turn $\mathrm{c \cdots cqq}$-QAM proof system
that on input~${(x,r)}$ simulates the last ${m-1}$~turns of $\Pi$ on $x$
under the condition that the first message from the verifier was $r$ in $\Pi$.
Let ${B = (B_\yes, B_\no)}$ be the following promise problem
in ${\mathrm{c \cdots c}\qqQAM(m-1)}$:
\begin{alignat*}{2}
&
B_\yes
&&
=
\set{(x,r)}{\mbox{the maximum acceptance probability in $\Pi^{(-1)}$ on input~${(x,r)}$ is at least $2/3$}},
\\
&
B_\no
&&
=
\set{(x,r)}{\mbox{the maximum acceptance probability in $\Pi^{(-1)}$ on input~${(x,r)}$ is at most $1/3$}}.
\end{alignat*}
Note that, if $x$ is in $A_\yes$,
then ${(x,r)}$ is in $B_\yes$ for at least ${(1 - 3 \cdot 2^{-8})}$-fraction of the choices of $r$.
Similarly, if $x$ is in $A_\no$,
then ${(x,r)}$ is in $B_\no$ for at least ${(1 - 3 \cdot 2^{-8})}$-fraction of the choices of $r$.
By the result for the case with odd~$m$ above,
it holds that $B$ is in ${\mathrm{c \cdots c}\qqQAM(m-2)}$.
Thus, there exists an ${(m-2)}$-turn $\mathrm{c \cdots cqq}$-QAM proof system~$\Pi^{(-2)}$ for $B$
such that if ${(x,r)}$ is in $B_\yes$,
the maximum acceptance probability in $\Pi^{(-2)}$ is at least $2/3$,
while if ${(x,r)}$ is in $B_\no$,
the maximum acceptance probability in $\Pi^{(-2)}$ is at most $1/3$.
Note that the first turn of $\Pi^{(-2)}$ is a turn for the verifier,
and thus, one can merge the turn for sending $r$ with the first turn of $\Pi^{(-2)}$.
This results in an ${(m-2)}$-turn $\mathrm{c \cdots cqq}$-QAM proof system~$\Pi''$ for $A$
in which at the first turn
the new verifier sends a string ${r \in \Binary^{l(\abs{x})}}$ chosen uniformly at random
in addition to the original first message of the verifier in $\Pi^{(-2)}$ on input~${(x,r)}$,
and then behaves exactly in the same manner
as the verifier in $\Pi^{(-2)}$ on input~${(x,r)}$ in the rest of the protocol.
If $x$ is in $A_\yes$,
the maximum acceptance probability in this $\Pi''$ is at least
${(1 - 3 \cdot 2^{-8}) \cdot (2/3) > 5/8}$,
while if $x$ is in $A_\no$,
the maximum acceptance probability in $\Pi''$ is at most
${3 \cdot 2^{-8}+(1 - 3\cdot 2^{-8}) \cdot (1/3) < 3/8}$,
which is sufficient for the claim, due to Lemma~\ref{Lemma: amplification}. 
\end{proof}

Second, using the fact that $\CITM$ is $\qqQAM$-complete,
it is proved that ${\mathrm{c}\qqQAM \subseteq \qqQAM}$.

\begin{lemma}
${\mathrm{c}\qqQAM \subseteq \qqQAM}$.
\label{Lemma: cqq-QAM is in qq-QAM}
\end{lemma}

\begin{proof}
Let ${A = (A_\yes,A_\no)}$ be a problem in ${\mathrm{c}\qqQAM}$.
Then, $A$ has a $\mathrm{cqq}$-QAM proof system~$\Pi$ with completeness~$2/3$ and soundness~$1/3$.
Let $l$ be the polynomially bounded function that specifies the length of the first message in $\Pi$.
Consider the $\mathrm{qq}$-QAM proof system~$\Pi^{\mathrm{qq}}$
that on input~${(x,w)}$ simulates the last two turns of $\Pi$ on $x$
under the condition that the first message in $\Pi$ was $w\in\Binary^{l(\abs{x})}$.
Let ${B = (B_\yes, B_\no)}$ be the following promise problem in $\qqQAM$:
\begin{alignat*}{2}
&
B_\yes
&&
=
\set{(x,w)}{\mbox{the maximum acceptance probability in $\Pi^{\mathrm{qq}}$ on input~${(x,w)}$ is at least $2/3$}},
\\
&
B_\no
&&
=
\set{(x,w)}{\mbox{the maximum acceptance probability in $\Pi^{\mathrm{qq}}$ on input~${(x,w)}$ is at most $1/3$}}.
\end{alignat*}
Note that for any $x$,
if $x$ is in $A_\yes$, there exists a string~$w$ in $\Binary^{l(\abs{x})}$
such that ${(x,w)}$ is in $B_\yes$,
and if $x$ is in $A_\no$, for every string~$w$ in $\Binary^{l(\abs{x})}$,
${(x,w)}$ is in $B_\no$.

Let $\function{p}{\Nonnegative}{\Natural}$ be a non-decreasing polynomially bounded function,
which will be fixed later.
First notice that $B$ has a $\mathrm{qq}$-QAM proof system
that satisfies completeness~${1 - 2^{-p}}$ and soundness~$2^{-p}$
(the existence of such a proof system is ensured by Lemma~\ref{Lemma: amplification}).
Starting from this $\mathrm{qq}$-QAM proof system,
the proof of Lemma~\ref{Lemma: qq-QAM-hardness of CITM} implies the existence of
a polynomial-time algorithm that, given ${(x,w)}$,
computes a description of a quantum circuit~$Q_{x,w}$
of ${q_\inp(\abs{x})}$~input qubits and ${q_\out(\abs{x})}$~output qubits
with the following properties:
\begin{itemize}
\item[(i)]
  if ${(x,w)}$ is in $B_\yes$,
  there exists a quantum state~$\rho$ consisting of ${q_\inp(\abs{x})}$~qubits such that
  ${D(Q_{x,w}(\rho), (I/2)^{\tensor q_\out(\abs{x})}) \leq 2^{-p(\abs{x} + \abs{w}) - 1} < 2^{-p(\abs{x})}}$, and
\item[(ii)]
  if ${(x,w)}$ is in $B_\no$,
  for any quantum state~$\rho$ consisting of ${q_\inp(\abs{x})}$~qubits,
  ${D(Q_{x,w}(\rho), (I/2)^{\tensor q_\out(\abs{x})}) > 1/20}$.
\end{itemize}
Let $q$ be another non-decreasing polynomially bounded function satisfying 
${q(n) \geq \max \{ l(n) + 4, n \}}$ for any $n$ in $\Nonnegative$.
Considering the quantum circuit~$Q'_{x,w}$
that runs ${k(\abs{x})}$~copies of $Q_{x,w}$ in parallel
for the polynomially bounded function~${k = \bigceil{\frac{2 \ln 2}{\ln (400/399)} q}}$
and taking ${p = q + \ceil{\log k}}$,
it follows from Lemma~\ref{Lemma: polarization of minimum output trace distance}
(with $\Phi$ being the transformation induced by $Q_{x,w}$
and $\Psi$ being the transformation that receives an input state of ${q_\inp(\abs{x})}$~qubits
and always outputs the totally mixed state~${(I/2)^{\tensor q_\out(\abs{x})}}$ regardless of the input)
that
\begin{itemize}
\item[(i)]
  if $x$ is in $A_\yes$,
  there exist a string~$w$ in $\Binary^{l(\abs{x})}$
  and a quantum state~$\rho'$ consisting of ${q'_\inp(\abs{x})}$~qubits
  such that ${D(Q'_{x,w}(\rho'), (I/2)^{\tensor q'_\out(\abs{x})}) < 2^{-q(\abs{x})}}$, and
\item[(ii)]
  if $x$ is in $A_\no$,
  for any string~$w$ in $\Binary^{l(\abs{x})}$
  and any quantum state~$\rho'$ consisting of ${q'_\inp(\abs{x})}$~qubits,
  ${D(Q'_{x,w}(\rho'), (I/2)^{\tensor q'_\out(\abs{x})}) > 1 - 2^{-q(\abs{x})}}$,
\end{itemize}
where ${q'_\inp = k q_\inp}$ and ${q'_\out = k q_\out}$.

Now consider the quantum circuit~$R_x$
of ${l(\abs{x}) + q'_\inp(\abs{x})}$~input qubits and ${q'_\out(\abs{x})}$~output qubits
that corresponds to the following algorithm:
\begin{step}
\item
  Measure all the ${l(\abs{x})}$ qubits in the quantum register~$\sfW$
  in computational basis to obtain a classical string~$w$ in $\Binary^{l(\abs{x})}$,
  where $\sfW$ corresponds to the first ${l(\abs{x})}$ qubits of the input qubits.
\item
  Compute from ${(x,w)}$ a description of the quantum circuit~$Q'_{x,w}$.
  Perform the circuit~$Q'_{x,w}$ with qubits in the quantum register~$\sfR$
  as its input qubits,
  where $\sfR$ corresponds to the last ${q'_\inp(\abs{x})}$~qubits of the input qubits of $R_x$.
  Output the qubits corresponding to the output qubits of $Q'_{x,w}$.
\end{step}
We claim that the circuit~$R_x$ satisfies the following two properties:
\begin{itemize}
\item[(i)]
  if $x$ is in $A_\yes$,
  there exists a quantum state~$\sigma$ consisting of ${l(\abs{x}) + q'_\inp(\abs{x})}$~qubits such that
  ${D(R_x(\sigma), (I/2)^{\tensor q'_\out(\abs{x})}) < 2^{-q(\abs{x})}}$,
  and
\item[(ii)]
  if $x$ is in $A_\no$,
  for any quantum state~$\sigma$ consisting of ${l(\abs{x}) + q'_\inp(\abs{x})}$~qubits,
  ${D(R_x(\sigma), (I/2)^{\tensor q'_\out(\abs{x})}) > 1/q'_\out(\abs{x}) }$.
\end{itemize}

In fact, the item~(i) is obvious from the construction of $R_x$.

For the item~(ii), suppose that $x$ is in $A_\no$.
Then,
for any string~$w$ in $\Binary^{l(\abs{x})}$
and any quantum state~$\rho'$ consisting of ${q'_\inp(\abs{x})}$~qubits,
it holds that 
${D(Q'_{x,w}(\rho'), (I/2)^{\tensor q'_\out(\abs{x})}) > 1 - 2^{-q(\abs{x})}}$.
From Lemma~\ref{Lemma: NC00Thm11.10} and the second inequality of Lemma~\ref{Lemma: trace distance and entropy},
it follows that
\[
S(R_x(\sigma))
<
l(\abs{x}) + q'_\out(\abs{x}) - q(\abs{x}) + 2
\leq
q'_\out(\abs{x}) - 2
\leq
\biggl( 1 - \frac{1}{q'_\out(\abs{x})} - 2^{-q'_\out(\abs{x})} \biggr) q'_\out(\abs{x}).
\]
Hence, the first inequality of Lemma~\ref{Lemma: trace distance and entropy} ensures that
${D(R_x(\sigma), (I/2)^{\tensor q'_\out(\abs{x})}) > 1/q'_\out(\abs{x})}$.

Finally, consider the quantum circuit~$R'_x$ that runs ${k'(\abs{x})}$~copies of $R_x$ in parallel 
for a polynomially bounded function~${k' = \bigceil{\frac{2\ln(1/2)}{\ln(1-(1/(q'_\out)^2))}} \leq 2(q'_\out)^2}$. 
Assuming that ${q'_\out(\abs{x})^2 \leq 2^{q(\abs{x})-4}}$
(otherwise $\abs{x}$ is at most some fixed constant,
as $q'_\out$ is a polynomially bounded function and ${q(\abs{x}) \geq \abs{x}}$,
and thus,
it can be checked trivially whether $x$ is in $A_\yes$ or in $A_\no$),
it follows from Lemma~\ref{Lemma: polarization of minimum output trace distance} that
\begin{itemize}
\item[(i)]
  if $x$ is in $A_\yes$,
  there exists a quantum state~$\sigma$ consisting of $q''_\inp(\abs{x})$~qubits
  such that ${D(R'_x(\sigma), (I/2)^{\tensor q''_\out(\abs{x})}) < 1/8}$,
  and
\item[(ii)]
  if $x$ is in $A_\no$,
  for any quantum state~$\sigma$ consisting of $q''_\inp(\abs{x})$~qubits,
 ${D(R'_x(\sigma), (I/2)^{\tensor q''_\out(\abs{x})}) > 1/2 }$,
\end{itemize}
where ${q''_\inp = k'(l + q'_\inp)}$ and ${q''_\out = k'(l + q'_\out)}$.
Therefore,
$R'_x$ is a yes-instance of ${\CITM(1/8, 1/2)}$ if $x$ is in $A_\yes$,
while $R'_x$ is a no-instance of ${\CITM(1/8, 1/2)}$ if $x$ is in $A_\no$.
This implies that
any problem~$A$ in ${\mathrm{c}\qqQAM}$ is reducible to ${\CITM(1/8, 1/2)}$ in polynomial time,
and thus in $\qqQAM$ by Lemma~\ref{Lemma: CITM is in qq-QAM},
which completes the proof.
\end{proof}

\begin{remark}
The proof of Lemma~\ref{Lemma: cqq-QAM is in qq-QAM}
essentially shows the $\qqQAM$-hardness of the \problemfont{Maximum Output Quantum Entropy Approximation} ($\MaxOutQEA$) problem.
On the other hand, the fact that $\MaxOutQEA$ is in $\qqQAM$
is easily proved by an almost straightforward modification
of the arguments in Refs.~\cite{BenSchTaS10ToC, ChaCioKerVad08TCC}
used to show that the \problemfont{Quantum Entropy Approximation} ($\QEA$) problem is in $\NIQSZK$.
Hence, the $\MaxOutQEA$ problem is also $\qqQAM$-complete,
proving Theorem~\ref{Theorem: qq-QAM-completeness of MaxOutQEA}.
A rigorous proof of Theorem~\ref{Theorem: qq-QAM-completeness of MaxOutQEA}
will be presented in the appendix.
\end{remark}

Finally, using Lemma~\ref{Lemma: cqq-QAM is in qq-QAM},
it is proved that ${\mathrm{cc}\qqQAM \subseteq \qqQAM}$.

\begin{lemma}
${\mathrm{cc}\qqQAM \subseteq \qqQAM}$.
\label{Lemma: ccqq-QAM is in qq-QAM}
\end{lemma}

\begin{proof}
Let ${A = (A_\yes, A_\no)}$ be a problem in ${\mathrm{cc}\qqQAM}$.
By Lemma~\ref{Lemma: amplification},
one can assume that $A$ has a $\mathrm{ccqq}$-QAM proof system~$\Pi$
with completeness~${1 - 2^{-8}}$ and soundness~$2^{-8}$.
Let $\Pi^{(-1)}$ be the $\mathrm{cqq}$-QAM proof system
that on input~${(x,r)}$ simulates the last three turns of $\Pi$ 
on input $x$ assuming that the first message in $\Pi$ from the verifier was~$r$.
Let ${B = (B_\yes, B_\no)}$ be the following promise problem in ${\mathrm{c}\qqQAM}$:
\begin{alignat*}{2}
&
B_\yes
&&
=
\set{(x,r)}{\mbox{the maximum acceptance probability in $\Pi^{(-1)}$ on input~${(x,r)}$ is at least $2/3$}},
\\
&
B_\no
&&
=
\set{(x,r)}{\mbox{the maximum acceptance probability in $\Pi^{(-1)}$ on input~${(x,r)}$ is at most $1/3$}}.
\end{alignat*}
Note that,
if $x$ is in $A_\yes$, then ${(x,r)}$ is in $B_\yes$ for at least ${(1 - 3 \cdot 2^{-8})}$-fraction of the choices of $r$, while if $x$ is in $A_\no$,
then ${(x,r)}$ is in $B_\no$ for at least ${(1 - 3 \cdot 2^{-8})}$-fraction of the choices of $r$.
By Lemma~\ref{Lemma: cqq-QAM is in qq-QAM}, it holds that $B$ is in $\qqQAM$.
Thus, there exists a $\mathrm{qq}$-QAM proof system~$\Pi'$ for $B$ such that,
if ${(x,r)}$ is in $B_\yes$,
the maximum acceptance probability in $\Pi'$ is at least $2/3$,
and if ${(x,r)}$ is in $B_\no$,
the maximum acceptance probability in $\Pi'$ is at most $1/3$.
Here, the first turn of $\Pi'$ is a turn for the verifier,
and thus, one can merge the turn for sending $r$ with the first turn of $\Pi'$.
This results in another $\mathrm{qq}$-QAM proof system~$\Pi''$ for $A$
in which at the first turn
the new verifier sends a string ${r \in \Binary^{l(\abs{x})}}$ chosen uniformly at random
in addition to the original first message of the verifier in $\Pi'$ on input~${(x,r)}$,
and then behaves exactly in the same manner
as the verifier in $\Pi'$ on input~${(x,r)}$ in the rest of the protocol.
Notice that sending a random string~$r$ of length~${l(\abs{x})}$ can be exactly simulated by
sending the halves of ${l(\abs{x})}$~EPR pairs and measuring in the computational basis
all the remaining halves of them that the verifier possesses.
If $x$ is in $A_\yes$,
the maximum acceptance probability in this $\Pi''$ is at least
${(1 - 3 \cdot 2^{-8}) \cdot (2/3) > 5/8}$,
while if $x$ is in $A_\no$,
the maximum acceptance probability in $\Pi''$ is at most
${3 \cdot 2^{-8}+(1 - 3\cdot 2^{-8}) \cdot (1/3) < 3/8}$,
which is sufficient for the claim, due to Lemma~\ref{Lemma: amplification}. 
\end{proof}

Now one inclusion of Theorem~\ref{Theorem: c...cqq-QAM(m) = qq-QAM} is immediate
from Lemmas~\ref{Lemma: c...cqq-QAM(m) is in ccqq-QAM}~and~\ref{Lemma: ccqq-QAM is in qq-QAM}, and the other inclusion is trivial,
which completes the proof of Theorem~\ref{Theorem: c...cqq-QAM(m) = qq-QAM}.

Notice that all the proofs of Lemmas~\ref{Lemma: c...cqq-QAM(m) is in ccqq-QAM},~\ref{Lemma: cqq-QAM is in qq-QAM},~and~\ref{Lemma: ccqq-QAM is in qq-QAM}
can be easily modified to preserve the perfect completeness property.
Indeed, the proof of Lemma~\ref{Lemma: c...cqq-QAM(m) is in ccqq-QAM}
can be modified to preserve the perfect completeness property
by taking $B_\yes$ to be the set of ${(x,r)}$'s such that
the maximum acceptance probability in $\Pi^{(-1)}$ on input~${(x,r)}$ is one,
and using Lemma~\ref{Lemma: amplification for perfect completeness case}
instead of Lemma~\ref{Lemma: amplification}.
With a similar modification to the set~$B_\yes$
as well as using Corollary~\ref{Corollary: qq-QAM_1-completeness of CITM(0,b)}
instead of Theorem~\ref{Theorem: qq-QAM-completeness of CITM},
the proof of Lemma~\ref{Lemma: cqq-QAM is in qq-QAM} can be modified to present
a reduction from any problem in ${\mathrm{c}\qqQAM_1}$ to ${\CITM(0,b)}$,
which shows the inclusion~${\mathrm{c}\qqQAM_1 \subseteq \qqQAM_1}$.
Using this inclusion instead of Lemma~\ref{Lemma: cqq-QAM is in qq-QAM}
and again with a similar modification to $B_\yes$
and a replacement of Lemma~\ref{Lemma: amplification} by Lemma~\ref{Lemma: amplification for perfect completeness case},
the proof of Lemma~\ref{Lemma: cqq-QAM is in qq-QAM} can be modified so that
it shows the inclusion~${\mathrm{cc}\qqQAM_1 \subseteq \qqQAM_1}$.
Hence, the following corollary holds.

\begin{corollary}
For any constant~${m \geq 2}$,
${
\mathrm{c \cdots c}\qqQAM_1(m) = \qqQAM_1
}$.
\label{Corollary: c...cqq-QAM_1(m) = qq-QAM_1}
\end{corollary}

% ---------------------------------------------------------------------------
%   QAM versus one-sided error qqQAM
% ---------------------------------------------------------------------------

\section{$\boldsymbol{\QAM}$ versus One-Sided Error $\boldsymbol{\qqQAM}$}
\label{Section: QAM is in qq-QAM_1}

This section shows that $\mathrm{qq}$-QAM proof systems of perfect-completeness
are already as powerful as the standard QAM proof systems of two-sided bounded error (Theorem~\ref{Theorem: QAM is in qq-QAM_1}).
As mentioned at the end of Section~\ref{Section: c...cqq-QAM(m) = qq-QAM},
the collapse theorem for $\qqQAM$ holds even for the perfect-completeness variants. 
In particular, the inclusion~%
${
\mathrm{cc}\qqQAM_1 \subseteq \qqQAM_1
}$
holds.
Hence, for the proof of Theorem~\ref{Theorem: QAM is in qq-QAM_1},
it suffices to show that any problem in $\cqQAM$ (${= \QAM}$)
is necessarily in the class~${\mathrm{cc}\qqQAM_1}$.
As mentioned earlier, 
this can be shown by combining the classical technique
in Ref.~\cite{Cai12LectureNotes} for proving ${\AM = \AM_1}$,
which originates in the proof of ${\BPP \subseteq \Sigma_2^p}$ due to Lautemann~\cite{Lau83IPL},
and the recent result that sharing a constant number of EPR pairs can make
QMA proofs perfectly complete~\cite{KobLeGNis13ITCS}.

Intuitively, with two classical turns of communications,
the classical technique in Ref.~\cite{Cai12LectureNotes}
can be used to generate polynomially many instances of a (promise) QMA problem
such that all these instances are QMA yes-instances if the input was a yes-instance,
while at least one of these instances is a QMA no-instance with high probability
if the input was a no-instance
(some of the QMA instances may violate the promise if the input was a no-instance,
but this does not matter,
as the important point is that at least one instance is a no-instance in this case).
Now one makes use of 
the $\textEPRQMA{\const}_1$ proof system in Ref.~\cite{KobLeGNis13ITCS}
for each QMA instance,
by running polynomially many attempts of such a system in parallel
to see that none of them results in rejection. 
The resulting proof system is thus of $\mathrm{ccqq}$-QAM type,
as $\textEPRQMA{\const}_1$ proof systems are special cases of $\mathrm{qq}$-QAM proof systems.
The perfect completeness of this proof system follows from the fact
that all the QMA instances generated from an input of yes-instance are QMA yes-instances,
and all of them are accepted without error in the attempts of the $\textEPRQMA{\const}_1$ system
due to the perfect completeness property of the system.
The soundness of this proof system follows from the fact
that at least one QMA instance generated from an input of no-instance is a QMA no-instance with high probability,
for which the $\textEPRQMA{\const}_1$ proof system results in rejection with reasonably high probability,
due to the soundness property of it.

The rigorous proof
will use the following notion of fat and thin subsets of $\Binary^l$.
A subset~$S$ of $\Binary^l$ is \emph{fat}
if ${\frac{\abs{S}}{2^l} \geq 1 - \frac{1}{l}}$,
and is \emph{thin}
if ${\frac{\abs{S}}{2^l} \leq \frac{1}{l}}$.
For any ${S \subseteq \Binary^l}$ and ${r \in \Binary^l}$,
let ${S \xor r = \set{x \xor r}{x \in S}}$,
where for any $x$ and $y$ in $\Binary^l$,
${x \xor y}$ denotes a string in $\Binary^l$
obtained by taking the bitwise exclusive-OR of $x$ and $y$.
The following property holds (see Lemma~5.15 of Ref.~\cite{Cai12LectureNotes}). 

\begin{lemma}\label{Cai12Lem5.15}
For any positive integer~$l$ and any subset~$S$ of $\Binary^l$,
\begin{itemize}
\item[(i)]
  if $S$ is fat, for any positive integers~$k$~and~$l$ such that ${k < l}$,
  ${\Pr_{r_1, \ldots, r_k \in \Binary^l} \bigl[ \bigcap_{j=1}^k (S \xor r_j) \neq \emptyset \bigr] = 1}$, and
\item[(ii)]
  if $S$ is thin, for any positive integer~$k$,
  ${\Pr_{r_1, \ldots, r_k \in \Binary^l} \bigl[ \bigcap_{j=1}^k (S \xor r_j) = \emptyset \bigr] \geq 1 - \frac{2^l}{l^k}}$.
\end{itemize}
\end{lemma}

Using this lemma, Theorem~\ref{Theorem: QAM is in qq-QAM_1} is proved as follows.

\begin{proof}[Proof of Theorem~\ref{Theorem: QAM is in qq-QAM_1}] 
Let ${A = (A_\yes, A_\no)}$ be a problem in $\cqQAM$ (${= \QAM}$).
By Lemma~\ref{Lemma: amplification}, $A$ has a $\mathrm{cq}$-QAM proof system $\Pi$
with completeness~${1 - \frac{1}{3l}}$ and soundness~$\frac{1}{3l}$,
where $l$ is the polynomially bounded function
that specifies the length of the random string sent by the verifier at the first turn
(such a proof system indeed exists, as one can achieve exponentially small completeness and soundness errors if one likes,
while the message length remain polynomially bounded even in such cases).
Let $V$ denote the verifier in this system~$\Pi$.
Without loss of generality, one can assume that ${l \geq 4}$,
and $l$ also specifies the number of qubits $V$ would receive at the last turn in $\Pi$.
Consider the QMA proof system~$\Pi^\QMA$
that on input~${(x,r)}$ simulates the last turn of $\Pi$ on $x$
assuming that the first message in $\Pi$ from the verifier was $r$
(i.e., on input~${(x,r)}$, the verifier in $\Pi^\QMA$ first receives a quantum witness of ${l(\abs{x})}$~qubits,
and then simulates the final verification procedure of $V$ in $\Pi$ on input~$x$ conditioned that $V$ sent $r$ as his/her  
question at the first turn).
Let ${B = (B_\yes, B_\no)}$ be the following promise problem in $\QMA$:
\begin{alignat*}{2}
&
B_\yes
&&
=
\set{(x,r)}{\mbox{the maximum acceptance probability in $\Pi^\QMA$ on input ${(x,r)}$ is at least $2/3$}},
\\
&
B_\no
&&
=
\set{(x,r)}{\mbox{the maximum acceptance probability in $\Pi^\QMA$ on input ${(x,r)}$ is at most $1/3$}}.
\end{alignat*}
Note that,
if $x$ is in $A_\yes$,
then ${(x,r)}$ is in $B_\yes$ for at least ${(1 - \frac{1}{l(\abs{x})})}$-fraction of the choices of ${r \in \Binary^{l(\abs{x})}}$,
while if $x$ is in $A_\no$,
then ${(x,r)}$ is in $B_\no$ for at least ${(1 - \frac{1}{l(\abs{x})})}$-fraction of the choices of ${r \in \Binary^{l(\abs{x})}}$.

Consider another $\mathrm{cq}$-QAM proof system~$\Pi'$ specified by the following protocol of the verifier on input~$x$:

\begin{step}
\item
  Send ${(l(\abs{x}) - 1)}$~strings~${r_1, \ldots, r_{l(\abs{x})-1}}$,
  each chosen uniformly at random from $\Binary^{l(\abs{x})}$.
\item
  Upon receiving a string~$r$ in $\Binary^{l(\abs{x})}$
  as well as ${(l(\abs{x}) - 1)}$~quantum registers~${\sfM_1, \ldots, \sfM_{l(\abs{x})-1}}$ of ${l(\abs{x})}$~qubits, 
  simulate the final verification procedure of $V$ in the original system~$\Pi$ on input~$x$
  with the question~${r \oplus r_j}$ and the quantum state in $\sfM_j$
  for each $j$ in ${\{1, \ldots, l(\abs{x})-1 \}}$
  (i.e., for each $j$,
  simulate the QMA proof system~$\Pi^\QMA$ on instance~${(x, r \xor r_j)}$
  with the quantum state in $\sfM_j$ as its quantum witness).
  Accept if and only if all the ${(l(\abs{x}) - 1)}$~simulations result in the acceptance.
\end{step}
The key point is that, if $x$ is in $A_\yes$,
for any choice of ${(r_1, \ldots, r_{l(\abs{x})-1})}$, 
there always exists an $r$ in $\Binary^{l(\abs{x})}$ such that
the pair~${(x, r \xor r_j)}$ is in $B_\yes$
for \emph{all} $j$ in ${\{1, \ldots, l(\abs{x})-1 \}}$.
Indeed, if $x$ is in $A_\yes$, the set~$S_x^\yes$ defined by
\[
S_x^\yes
=
\set{r \in \Binary^{l(\abs{x})}}{(x,r) \in B_\yes}
\]
is fat, and hence by Lemma~\ref{Cai12Lem5.15},
for any ${r_1, \ldots, r_{l(\abs{x})-1}}$ in ${\Binary^{l(\abs{x})}}$,
there exists an $r$ in $\Binary^{l(\abs{x})}$
such that, for every $j$ in ${\{1, \ldots, l(\abs{x}) - 1\}}$,
the pair~${(x, r \xor r_j)}$ is in $B_\yes$. 

If $x$ is in $A_\no$, on the other hand,
it happens with very small probability that there exists an $r$ such that,
for all $j$, the QMA instance~${(x, r \xor r_j)}$ has
maximum acceptance probability greater than $1/3$
(here one must be a bit careful, because there may be QMA instances breaking the promise,
which is why the condition ``greater than $1/3$'' is used instead of ``at least $2/3$'').
This means that, if $x$ is in $A_\no$,
with very high probability over the choices of ${(r_1, \ldots, r_{l(\abs{x})-1})}$,
for any $r$ given, there exists at least one $j$ such that ${(x, r \xor r_j)}$ is in $B_\no$.
Indeed, if $x$ is in $A_\no$, the set~$S_x^{\neg\no}$ defined by
\[
S_x^{\neg\no}
=
\set{r \in \Binary^{l(\abs{x})}}{(x,r) \not\in B_\no}
\]
is thin, and hence by Lemma~\ref{Cai12Lem5.15},
the probability over the choices of ${(r_1, \ldots, r_{l(\abs{x})-1})}$ that
for every ${r \in \Binary^{l(\abs{x})}}$
there exists an index~$j$ in ${\{1, \ldots, l(\abs{x}) - 1\}}$ such that
the pair~${(x, r \xor r_j)}$ is in $B_\no$
is at least ${1 - \frac{2^{l(\abs{x})}}{l(\abs{x})^{l(\abs{x})-1}} \geq 1 - 2 ^{-l(\abs{x}) + 2}}$.

Finally, consider the following $\mathrm{ccqq}$-QAM proof system~$\Pi''$
that plugs in the idea of Ref.~\cite{KobLeGNis13ITCS} into each instance 
${(x, r \xor r_j)}$ of the (promise) QMA problem:
The verifier basically simulates $\Pi'$,
except that now, instead of ${\Pi}^\QMA$, 
he/she performs the $\textEPRQMA{\const}_1$ protocol
(Fig.~6 in Ref.~\cite{KobLeGNis13ITCS})
for each QMA instances.
For this, in addition to $r$ and ${\sfM_1, \ldots, \sfM_{l(\abs{x})-1}}$,
the verifier receives polynomially many single-qubit registers,
assuming that the verifier and prover share that polynomially many number of EPR pairs beforehand
-- these EPR pairs can be shared by adding a quantum turn for the verifier
after having received the response~$r$ from the prover.
Here note that one needs only a constant number of EPR pairs
for each instance~${(x, r \xor r_j)}$,
but one needs them for all ${(l(\abs{x}) - 1)}$~instances ${(x, r \xor r_j)}$,
which results in polynomially many EPR pairs in total.
Figure~\ref{Figure: verifier's ccqq-QAM_1 protocol for QAM}
presents a more precise description of the protocol for the verifier in the $\mathrm{ccqq}$-QAM proof system~$\Pi''$.

\begin{figure}[t!]
\begin{protocol*}{Verifier's $\boldsymbol{\mathrm{ccqq}}$-$\boldsymbol{\textrm{QAM}_1}$ Protocol for $\boldsymbol{\QAM}$}
\begin{step}
\item
  Send ${(l(\abs{x}) - 1)}$~strings~${r_1, \ldots, r_{l(\abs{x})-1}}$,
  each chosen uniformly at random from $\Binary^{l(\abs{x})}$,
  to the prover.
\item
  Receive a string~$r$ in $\Binary^{l(\abs{x})}$ from the prover.
  Prepare ${N (l(\abs{x}) - 1)}$~pairs of single-qubit registers~${(\sfS_{j,k}, \sfS'_{j,k})}$
  for each $j$ in ${\{1, \ldots, l(\abs{x})-1\}}$ and $k$ in ${\{1, \ldots, N\}}$,
  and generate an EPR pair in each of ${(\sfS_{j,k}, \sfS'_{j,k})}$, 
  where $N$ is the constant such that
  $N$~shared EPR pairs can make any QMA proof system perfectly complete
  in the construction of Ref.~\cite{KobLeGNis13ITCS}.
  Send each $\sfS'_{j,k}$ to the prover.
\item
  Receive $\sfM_j$ and ${\sfS'_{j,1}, \ldots, \sfS'_{j,N}}$ from the prover,
  for each $j$ in ${\{1, \ldots, l(\abs{x})-1\}}$. 
  Perform the verification procedure in the construction of Ref.~\cite{KobLeGNis13ITCS}
  for each QMA instance~${(x, r \xor r_j)}$, ${j \in \{1, \ldots, l(\abs{x})-1 \}}$,
  using $\sfM_j$ and ${(\sfS_{j,1}, \sfS'_{j,1}), \ldots, (\sfS_{j,N}, \sfS'_{j,N})}$.
  Accept if all the verification procedures result in acceptance,
  and reject otherwise.
\end{step}
\end{protocol*}
\caption{Verifier's $\mathrm{ccqq}$-QAM protocol for achieving perfect completeness for the problems in $\QAM$.}
\label{Figure: verifier's ccqq-QAM_1 protocol for QAM}
\end{figure}

This proves that $A$ is in ${\mathrm{cc}\qqQAM_1}$:
If $x$ is in $A_\yes$, for every choice of ${(r_1, \ldots, r_{l(\abs{x})-1})}$,
the verifier of $\Pi''$ always accepts
due to the perfect completeness of the $\textEPRQMA{\const}_1$ proof system.
If $x$ is in $A_\no$, the verifier can reject with reasonably high probability,
since it is guaranteed by the soundness of the $\textEPRQMA{\const}_1$ proof system
that the verifier of $\Pi''$ can detect a no-instance~${(x, r \xor r_j)}$ of the QMA problem with reasonably high probability,
and at least one such no-instance exists
with probability at least ${1 - 2^{-l(\abs{x})+2}}$
over the choices of ${(r_1, \ldots, r_{l(\abs{x})-1})}$. 
As Corollary~\ref{Corollary: c...cqq-QAM_1(m) = qq-QAM_1}
in particular ensures that ${\mathrm{cc}\qqQAM_1 \subseteq \qqQAM_1}$,
it follows that $A$ is in $\qqQAM_1$, as claimed.
\end{proof}

The fact that perfect completeness is achievable in $\ccQAM$
(Theorem~\ref{Theorem: ccQAM = ccQAM_1})
can be proved in a similar fashion,
except that now one uses the fact ${\MQA = \MQA_1}$
(a.k.a., ${\QCMA = \QCMA_1}$)
that any classical-witness QMA proofs
can be made perfectly complete shown in Ref.~\cite{JorKobNagNis12QIC}
instead of the inclusion~${\QMA \subseteq \EPRQMA{\const}_1}$.
Each QMA instance in the argument above
are replaced by an MQA (QCMA) instance in this case.
Notice that no additional turn is necessary in this case,
and the resulting proof system corresponding to $\Pi''$ 
is immediately a $\textrm{cc}$-QAM proof system of perfect completeness.

% ---------------------------------------------------------------------------
%   Collapse Theorem for General Quantum Arthur-Merlin Proofs
% ---------------------------------------------------------------------------

\section{Collapse Theorem for General Quantum Arthur-Merlin Proof Systems}
\label{Section: collapse theorem for general QAM}

Before the proof of Theorem \ref{Theorem: complete classifications}, first observe the simple fact that
one can always replace classical turns by quantum ones without diminishing the verification power,
which can be shown as in the proof of Lemma~\ref{Lemma: amplification}
by letting the verifier simulate classical turns by quantum turns via CNOT applications.

\begin{proposition}
\label{Proposition: c->q}
For any constant~$m$ in $\Natural$, any $j$ in ${\{1, \ldots, m\}}$, and any ${t_1, \ldots, t_m}$ in ${\{\rmc, \rmq\}}$,
\[
\GQAM{t_m \cdots t_{j+1} \, t_j \, t_{j-1} \cdots t_1}(m)
\subseteq
\GQAM{t_m\cdots t_{j+1} \, \rmq \, t_{j-1} \cdots t_1}(m).
\]  
\end{proposition}

As generalized quantum Arthur-Merlin proofs are nothing but a special case of general quantum interactive proofs,
it is obvious that
for any constant~$m$ and any ${t_1, \ldots, t_m}$ in ${\{\rmc, \rmq\}}$,
${\GQAM{t_m \cdots t_1}(m)}$ is contained in ${\QIP = \PSPACE}$~\cite{JaiJiUpaWat11JACM}.
As mentioned in Section~\ref{Section: introduction},
Marriott~and~Watrous~\cite{MarWat05CC} proved that ${\mathrm{q}\cqQAM}$ (${= \QMAM}$) already hits the ceiling, i.e., coincides with $\QIP$.
Next lemma (Lemma~\ref{Lemma: qcc-QAM = QMAM}) states 
that one can slightly improve this
and even the third message is not necessary to be quantum 
to have the whole power of general quantum interactive proofs. 
The proof is based on a simulation of the original $\mathrm{qcq}$-QAM system 
by a $\mathrm{qcc}$-QAM system using quantum teleportation.

\begin{lemma}
${\mathrm{q}\cqQAM \subseteq \mathrm{q}\ccQAM}$.
\label{Lemma: qcc-QAM = QMAM}
\end{lemma}

\begin{proof}
Let ${A = (A_\yes, A_\no)}$ be a problem in ${\mathrm{q}\cqQAM}$,
meaning that $A$ has a $\mathrm{qcq}$-QAM proof system~$\Pi$
with completeness~$2/3$ and soundness~$1/3$
that is specified by the protocol of the verifier of the following form for every input~$x$:
\begin{step}
\item
  Receive a quantum register~$\sfM_1$ from the prover,
  and then send a random string~$r$ to the prover.
\item
  Receive a quantum register~$\sfM_2$ from the prover.
  Prepare a private quantum register~$\sfV$,
  and perform the final verification procedure over ${(\sfM_1, \sfM_2, \sfV)}$.
\end{step}
Let $l$ be the polynomially bounded function that specifies the number of qubits in $\sfM_2$.
Consider the teleportation-based simulation of $\Pi$ by the $\mathrm{qcc}$-QAM proof system~$\widetilde{\Pi}$
that is specified by the protocol of the verifier of the following form for every input~$x$:
\begin{step}
\item
  Receive a quantum register~$\sfS_1$ of ${l(\abs{x})}$~qubits,
  in addition to the quantum register~$\sfM_1$,
  from the prover.
  Send a random string~$r$ to the prover as would be done in $\Pi$.
\item
  Receive a binary string~$b$ of length~${2 l(\abs{x})}$ from the prover.
  Apply ${X^{b_{j,1}}Z^{b_{j,2}}}$ to the $j$th qubit of $\sfS_1$,
  for each $j$ in ${\{1, \ldots, l(\abs{x})\}}$,
  where $b_{j,1}$~and~$b_{j,2}$ denote the ${(2j-1)}$st and ${(2j)}$th bits of $b$, respectively.
  Finally, prepare his/her private quantum register~$\sfV$ as in $\Pi$,
  and simulate the final verification procedure of the verifier in $\Pi$
  with ${(\sfM_1, \sfS_1, \sfV)}$.
\end{step}

For the completeness, suppose that $x$ is in $A_\yes$.
Then there exists a prover~$P$ who makes the verifier accept 
with probability~${p \geq 2/3}$ in the original $\mathrm{qcq}$-QAM system~$\Pi$.
Without loss of generality, one can assume that
$P$ has quantum registers~$\sfM_1$,~$\sfM_2$,~and~$\sfP$ at the beginning of the protocol,
where $\sfP$ is the private quantum register of $P$.
Let $\rho_x$ be the quantum state $P$ prepares in ${(\sfM_1, \sfM_2, \sfP)}$
at the first turn in $\Pi$,
and let $P_{x,r}$ be the unitary transformation $P$ applies to ${(\sfM_2, \sfP)}$
at the third turn in $\Pi$ when $P$ has received $r$.

In the $\mathrm{qcc}$-QAM system~$\widetilde{\Pi}$,
let the prover~$\widetilde{P}$ behave as follows:
On input~$x$, $\widetilde{P}$ prepares quantum registers~$\sfS_1$~and~$\sfS_2$,
each of ${l(\abs{x})}$~qubits, 
in addition to $\sfM_1$,~$\sfM_2$,~and~$\sfP$.
$\widetilde{P}$ generates $\rho_x$ in ${(\sfM_1,\sfM_2,\sfP)}$,
and also generates $\ket{\Phi^+}^{\tensor l(\abs{x})}$ in ${(\sfS_1, \sfS_2)}$
so that the $j$th qubit of $\sfS_1$ and that of $\sfS_2$ form an EPR pair,
for every $j$ in ${\{1, \ldots, l(\abs{x})\}}$.
$\widetilde{P}$ then sends $\sfM_1$ and $\sfS_1$ to the verifier at the first turn.
Upon receiving~$r$, $\widetilde{P}$ first applies $P_{x,r}$ to ${(\sfM_2, \sfP)}$ as $P$ would do,
and then measures the $j$th pair of qubits in ${(\sfS_2, \sfM_2)}$ in the Bell basis
to obtain a two-bit outcome~$b_j$, 
for every $j$ in ${\{1, \ldots, l(\abs{x})\}}$,
where $b_j$ equals $00$, $01$, $10$, and $11$
if the measurement results in $\ket{\Phi^+}$, $\ket{\Phi^-}$, $\ket{\Psi^+}$, and $\ket{\Psi^-}$, respectively.
$\widetilde{P}$ sends a binary string~$b$ of length~${2 l(\abs{x})}$
such that the pair of the ${(2j - 1)}$st and ${(2j)}$th bits is exactly $b_j$,
for every $j$ in ${\{1, \ldots, l(\abs{x})\}}$.
This makes the quantum state in $\sfM_2$ be teleported to that in $\sfS_1$,
as the application of the Pauli operators in the final step of the verifier in $\widetilde{\Pi}$ 
correctly removes the phase and/or bit errors if exist.
Hence the verifier accepts in $\widetilde{\Pi}$
with exactly the same probability~$p$ as in $\Pi$,
which ensures the completeness of $\widetilde{\Pi}$.

For the soundness, suppose that $x$ is in $A_\no$.
Let $\widetilde{P}'$ be any prover in $\widetilde{\Pi}$.
Without loss of generality, one can assume that
$\widetilde{P}'$ has quantum registers~$\sfM_1$,~$\sfS_1$,~and~$\widetilde{\sfP}'$ at the beginning of the protocol,
where $\widetilde{\sfP}'$ is the private quantum register of $\widetilde{P}'$.
Let $\rho$ be the quantum state $\widetilde{P}'$ prepares in ${(\sfM_1, \sfS_1, \widetilde{\sfP}')}$
at the first turn in $\widetilde{\Pi}$,
and let ${\{ \widetilde{P}_{x,r}^b \}_{b \in \Binary^{2 l(\abs{x})}}}$ be
the ${2 l(\abs{x})}$-bit outcome measurement
that $\widetilde{P}'$ performs over $\widetilde{\sfP}'$
at the third turn in $\widetilde{\Pi}$,
when $\widetilde{P}'$ has received $r$.

In the $\mathrm{qcq}$-QAM system~$\Pi$,
let the prover~$P'$ behave as follows:
On input~$x$, $P'$ prepares quantum registers~$\sfM_1$,~$\sfS_1$,~and~$\widetilde{\sfP}'$,
and generates $\rho$ in ${(\sfM_1, \sfS_1, \widetilde{\sfP}')}$,
as $\widetilde{P}'$ would do in $\widetilde{\Pi}$.
$P'$ then sends $\sfM_1$ to the verifier at the first turn in $\Pi$.
Upon receiving~$r$, $P'$ first performs
the ${2 l(\abs{x})}$-bit outcome measurement~${\{ \widetilde{P}_{x,r}^b \}_{b \in \Binary^{2 l(\abs{x})}}}$
over $\widetilde{\sfP}'$ to obtain a ${2 l(\abs{x})}$-bit outcome~$b'$.
Let $b'_{j,1}$~and~$b'_{j,2}$ be the ${(2j-1)}$st and ${(2j)}$th bits of $b'$, respectively,
for each $j$ in ${\{1, \ldots, l(\abs{x})\}}$.
$P'$ then applies ${X^{b'_{j,1}}Z^{b'_{j,2}}}$ to the $j$th qubit of $\sfS_1$
for each $j$ in ${\{1, \ldots, l(\abs{x})\}}$,
as the verifier in $\widetilde{\Pi}$ would do,
and sends $\sfS_1$ to the verifier as the quantum register~$\sfM_2$. 
From the construction,
it is obvious that this $P'$ can make the verifier accept in $\Pi$
with exactly the same probability as $\widetilde{P}'$ could in $\widetilde{\Pi}$,
which must be at most $1/3$ from the soundness property of $\Pi$,
and the soundness of $\widetilde{\Pi}$ follows.
\end{proof}

With Lemma~\ref{Lemma: qcc-QAM = QMAM} in hand,
Theorem~\ref{Theorem: complete classifications} is proved as follows.

\begin{proof}[Proof of Theorem~\ref{Theorem: complete classifications}]
For the item~(i),
first notice that the inclusion~${\mathrm{q}\cqQAM \subseteq \mathrm{qc}\ccQAM}$
can be proved in a manner very similar to the proof of Lemma~\ref{Lemma: qcc-QAM = QMAM},
with not the honest prover but the verifier preparing the EPR pairs.
As ${\mathrm{q}\cqQAM = \QMAM = \QIP = \PSPACE}$,
together with Lemma~\ref{Lemma: qcc-QAM = QMAM},
this implies that ${\mathrm{qc}\ccQAM = \mathrm{q}\ccQAM = \PSPACE}$.
As adding more turns to ${\rmq \:\! t_3 t_2 t_1}$-QAM and ${\rmq \:\! t_2 t_1}$-QAM proof systems
does not diminish the verification power for any $t_1$, $t_2$, and $t_3$ in ${\{\rmq, \rmc\}}$,
this establishes the claim in the item~(i).

For the item~(ii),
again with a similar argument to the proof of Lemma~\ref{Lemma: qcc-QAM = QMAM},
it holds that, for any constant~${m \geq 2}$,
${\mathrm{c \cdots c}\qqQAM(m) \subseteq \mathrm{c \cdots c}\qcQAM(m)}$,
and thus, combined with Theorem~\ref{Theorem: c...cqq-QAM(m) = qq-QAM} and Proposition~\ref{Proposition: c->q},
the claim follows.

For the item~(iii),
it suffices to show that for any constant~${m \geq 3}$,
${\GQAM{\mathrm{c \cdots cq}}(m) \subseteq \GQAM{\mathrm{c \cdots cq}}(m-1)}$.
The case with ${m \geq 5}$ is proved with an argument similar to that in the proof of Lemma~\ref{Lemma: c...cqq-QAM(m) is in ccqq-QAM},
since the first three (resp. four) turns of the $m$-turn $\mathrm{c \cdots cq}$-QAM proof systems
are classical when $m$ is odd (resp. when $m$ is even).
In the case where ${m = 3}$,
one modifies the construction of $\Pi'$ in the proof of Lemma~\ref{Lemma: c...cqq-QAM(m) is in ccqq-QAM}
so that the message from the prover at the second turn (corresponding to Step~2 of $\Pi'$) is quantum,
consisting of two parts:
the $\sfY$ part and $\sfZ$ part,
each corresponding to $y$ and $z$ in Step~2 of $\Pi'$.
In order to force the content in the $\sfY$ part to be classical,
the verifier simply measures each qubit in the $\sfY$ part in the computational basis.
The analysis in the proof of Lemma~\ref{Lemma: c...cqq-QAM(m) is in ccqq-QAM} then works with the case where ${m = 3}$,
i.e., the case where a $\mathrm{ccq}$-QAM system is simulated by a $\mathrm{cq}$-QAM system.
The case where ${m=4}$ can then be proved using this result with ${m=3}$,
with the same argument as in the proof of Lemma~\ref{Lemma: c...cqq-QAM(m) is in ccqq-QAM}.

Finally, for the item~(iv),
it suffices to show that
the inclusion~${\GQAM{\mathrm{c \cdots c}}(m) \subseteq \GQAM{\mathrm{c \cdots c}}(m-1)}$
holds for any constant~${m \geq 3}$,
which easily follows from an argument similar to that in the proof of Lemma~\ref{Lemma: c...cqq-QAM(m) is in ccqq-QAM},
since all the messages are classical.
\end{proof}

% ---------------------------------------------------------------------------
%   Conclusion
% ---------------------------------------------------------------------------

\section{Conclusion}
\label{Section: conclusion}

This paper has introduced the generalized model of quantum Arthur-Merlin proof systems
to provide some new insights on the power of two-turn quantum interactive proofs.
A number of open problems are listed below concerning generalized quantum Arthur-Merlin proof systems and other related topics:
\begin{itemize}
\item
  Is there any natural problem, other than $\CITM$ and $\MaxOutQEA$, in $\qqQAM$
  that is not known to be in the standard $\QAM$?
  Or is $\qqQAM$ equal to $\QAM$?
\item
  Currently no upper-bound is known for $\qqQAM$ other than ${\QIP(2)}$.
  Can a better upper-bound be placed on $\qqQAM$?
  Is $\qqQAM$ contained in ${\BP \cdot \PP}$?
\item
  Does ${\qqQAM = \qqQAM_1}$?
  In other words, is perfect completeness achievable in $\qqQAM$?
  Similar questions remain open even for ${\QIP(2)}$ and $\QAM$.
\item
  What happens if some of the messages are restricted to be classical
  in the standard quantum interactive proof systems?
  Does a collapse theorem similar to the $\qqQAM$ case hold even with the ${\QIP(2)}$ case?
  More precisely, is the power of $m$-turn quantum interactive proof systems
  equivalent to ${\QIP(2)}$ for any constant~${m \geq 2}$,
  when the first ${(m-2)}$~turns are restricted to exchange only classical messages?
\end{itemize}
For the last question above,
note that one might be able to show a similar collapse theorem even with ${\QIP(2)}$
when the verifier \emph{cannot} use quantum operations at all during the first ${(m-2)}$~turns
(by extending the argument due to Goldwasser~and~Sipser~\cite{GolSip89ACR}
to replace the classical interaction of the first ${(m-2)}$~turns
by an $m$-turn classical public-coin interaction,
and then applying arguments similar to those in this paper,
using some appropriate ${\QIP(2)}$-complete problem like the \problemfont{Close Image} problem~\cite{Wat02QIP, HayMilWil12arXiv}).
A more difficult, but more natural and interesting case is
where the verifier can use quantum operations to generate his/her classical messages
even for the first ${(m-2)}$~turns,
to which the Goldwasser-Sipser technique does not seem to apply any longer.
A collapse theorem for such a case, if provable,
would be very helpful when trying to put more problems in ${\QIP(2)}$
and more generally investigating the properties of two-turn quantum interactive proof systems.

% ---------------------------------------------------------------------------
%   Acknowledgements
% ---------------------------------------------------------------------------

\subsection*{Acknowledgements}

The authors are grateful to Francesco Buscemi and Richard Cleve for very useful discussions.
This work is supported by
the Grant-in-Aid for Scientific Research~(A)~No.~24240001 of the Japan Society for the Promotion of Science
and the Grant-in-Aid for Scientific Research on Innovative Areas~No.~24106009 of
the Ministry of Education, Culture, Sports, Science and Technology in Japan.
HN also acknowledges support from
the Grant-in-Aids for Scientific Research~(A)~Nos.~21244007~and~23246071~and~(C)~No.~25330012 of the Japan Society for the Promotion of Science.

% ---------------------------------------------------------------------------
%   References
% ---------------------------------------------------------------------------

\newcommand{\etalchar}[1]{$^{#1}$}
\providecommand{\noopsort}[1]{}

% ---------------------------------------------------------------------------
%   Appendix
% ---------------------------------------------------------------------------

\appendix

\section{$\boldsymbol{\qqQAM}$-Completeness of \MaxOutQEA}

This section gives a rigorous proof of Theorem~\ref{Theorem: qq-QAM-completeness of MaxOutQEA}
that states that the $\MaxOutQEA$ problem is $\qqQAM$-complete.
First, it is proved that $\MaxOutQEA$ is in $\qqQAM$.

\begin{lemma}
$\MaxOutQEA$ is in $\qqQAM$.
\label{Lemma: MaxOutQEA is in qq-QAM}
\end{lemma}

\begin{proof}
We present a reduction from the $\MaxOutQEA$ problem to the $\CITM$ problem
(with some appropriate parameters),
by modifying the reduction from the $\QEA$ problem
to the \problemfont{Quantum State Closeness to Totally Mixed} ($\QSCTM$) problem
presented in Ref.~\cite{ChaCioKerVad07CePrint},
which relies on the analysis found in Section 5.3 of Ref.~\cite{BenSchTaS10ToC}.

Let ${x = (Q, t)}$ be an instance of $\MaxOutQEA$,
where $Q$ is a description of a quantum circuit that specifies a quantum channel~$\Phi$,
and $t$ is a positive integer.
For simplicity, in what follows, we identify the description~$Q$ and the quantum circuit it induces.
Suppose that $Q$ acts on $m_\all$~qubits with $m_\inp$~specified input qubits and $m_\out$~specified output qubits.
Let $q$ and $\varepsilon$ be two functions
that appear in Eqs.~(5.1) and (5.2) of Ref.~\cite{BenSchTaS10ToC}\footnotemark
to be specified later.
\footnotetext{
  Rigorously speaking, $q$ in the present case corresponds to $\frac{q}{2}$ in the left-hand sides of Eqs.~(5.1)~and~(5.2) in Ref.~\cite{BenSchTaS10ToC}.
  This is due to the fact
  that the $\MaxOutQEA$ problem in this paper is defined
  using threshold values~${t + 1}$~and~${t - 1}$,
  while the $\QEA$ problem in Ref.~\cite{BenSchTaS10ToC} is defined
  using threshold values~${t + \frac{1}{2}}$~and~${t - \frac{1}{2}}$.
}
We consider the quantum circuit~$Q^{\tensor q(\abs{x})}$ that runs ${q(\abs{x})}$~copies of $Q$ in parallel,
and the ${(qt, d, \varepsilon)}$-quantum extractor~$E$ on ${q(\abs{x}) m_\out}$~qubits given in Ref.~\cite[Section~5.3]{BenSchTaS10ToC},
which is written as ${E = \frac{1}{2^d} \sum_{i=1}^{2^d} E_i}$,
where ${E_i(\rho) = U_i \rho \conjugate{U_i}}$ for unitary operators~$U_i$.
Let $R$ be the quantum circuit 
that runs $Q^{\tensor q(\abs{x})}$
and then applies $E$ to the output state of ${q(\abs{x}) m_\out}$~qubits.
By following the analysis in Ref.~\cite{BenSchTaS10ToC},
one can show that
\begin{itemize}
\item[(i)]
  if ${x = (Q,t)}$ is a yes-instance of $\MaxOutQEA$,
  there exists a quantum state~$\rho$ of ${q(\abs{x}) m_\inp}$~qubits
  such that 
  ${D(R(\rho), (I/2)^{\tensor q(\abs{x}) m_\out}) \leq \frac{3}{2} \varepsilon}$,
  and
\item[(ii)]
  if ${x = (Q,t)}$ is a no-instance of $\MaxOutQEA$,
  for any quantum state~$\rho$ of ${q(\abs{x}) m_\inp}$~qubits,
  ${D(R(\rho), (I/2)^{\tensor q(\abs{x}) m_\out}) \geq \frac{1}{4q(\abs{x}) m_\out}}$.
\end{itemize}
In fact, the item~(i) follows from exactly the same analysis as in Ref.~\cite{BenSchTaS10ToC},
by taking ${\rho = \sigma^{\tensor q(\abs{x})}}$
with $\sigma$ being a quantum state of $m_\inp$~qubits
such that ${S(Q(\sigma)) \geq t + 1}$
(the condition~${S_\textmax(\Phi) \geq t + 1}$ ensures the existence of such a state~$\sigma$).

To prove the item~(ii),
first notice that,
if ${x = (Q,t)}$ is a no-instance of $\MaxOutQEA$,
it holds that
${S(Q(\sigma)) \leq S_\textmax(\Phi) \leq t - 1}$
for any quantum state~$\sigma$ of $m_\inp$~qubits.
Take an arbitrary quantum state~$\rho$ of ${q(\abs{x}) m_\inp}$~qubits.
By Lemma~\ref{Lemma: NC00Thm11.10},
it holds that
\[
S(R(\rho))
=
S \biggl(\frac{1}{2^d} \sum_{i=1}^{2^d} U_i Q^{\tensor q(\abs{x})}(\rho) \conjugate{U_i} \biggr)
\leq
S(Q^{\tensor q(\abs{x})}(\rho)) + d.
\]
For each $i$ in ${\{1, \ldots, q(\abs{x})\}}$,
let $\sfR_i$ be the output quantum register of the $i$th copy of $Q$
(hence, the whole output state~${Q^{\tensor q(\abs{x})}(\rho)}$ of $Q^{\tensor q(\abs{x})}$ is in ${(\sfR_1, \ldots, \sfR_{q(\abs{x})})}$),
and let $\sigma_{\sfR_i}$ be the reduced state of ${Q^{\tensor q(\abs{x})}(\rho)}$ of $m_\out$~qubits
obtained by tracing out all the qubits except those in $\sfR_i$.
By the subadditivity of von~Neumann entropy, it follows that
\[
S(Q^{\tensor q(\abs{x})}(\rho))
\leq
\sum_{i=1}^{q(\abs{x})} S(\sigma_{\sfR_i})
\leq
\sum_{i=1}^{q(\abs{x})} \max_\sigma S(Q(\sigma))
\leq
(t - 1) q(\abs{x}),
\]
which implies that
\[
S(R(\rho)) \leq (t-1)q(\abs{x}) + d.
\]  
Now the item~(ii) follows from exactly the same analysis as in Ref.~\cite{BenSchTaS10ToC}.

To complete the reduction,
similarly to Ref.~\cite{ChaCioKerVad07CePrint},
one takes ${\varepsilon = 1/2^k}$ for a polynomially bounded fuction~$k$
such that ${k(n) \geq n}$ for any $n$ in $\Nonnegative$ and ${k(n) \in O(n)}$,
and a polynomially bounded function~$q$
such that ${q(n) \in \Theta(n^4)}$
so that Eqs.~(5.1)~and~(5.2) are satisfied in Ref.~\cite{BenSchTaS10ToC}.
Consider the quantum circuit~$R'$ that runs ${r(\abs{x})}$~copies of $R$ in parallel
for a polynomially bounded function~$r$
such that
${r(n) = \bigceil{ \frac{2 \ln(1/2)}{\ln(1 - (1/(2q(n) m_\out)^2))}} \leq 2 (2 q(n) m_\out)^2}$
for all $n$ in $\Nonnegative$. 
Assuming that ${r(\abs{x}) \leq 2^{\abs{x}}/12}$
(otherwise $\abs{x}$ is at most some fixed constant,
as $r$ is a polynomially bounded function,
and thus,
it can be checked trivially whether ${x = (Q,t)}$ is a yes-instance or a no-instance),
it follows from Lemma~\ref{Lemma: polarization of minimum output trace distance} that
\begin{itemize}
\item[(i)]
  if ${x = (Q,t)}$ is a yes-instance,
  there exists a quantum state~$\sigma$ of ${r(\abs{x}) q(\abs{x}) m_\inp}$~qubits 
  such that ${D(R'(\sigma), (I/2)^{\tensor r(\abs{x}) q(\abs{x}) m_\out}) \leq 1/8}$,
  and
\item[(ii)]
  if ${x = (Q,t)}$ is a no-instance,
  for any quantum state~$\sigma$ of ${r(\abs{x}) q(\abs{x}) m_\inp}$~qubits,
  ${D(R'(\sigma), (I/2)^{\tensor r(\abs{x}) q(\abs{x}) m_\out}) \geq 1/2}$.
\end{itemize}
Hence, $\MaxOutQEA$ is reducible to ${\CITM(1/8, 1/2)}$ in polynomial time,
and thus in $\qqQAM$ by Lemma~\ref{Lemma: CITM is in qq-QAM}.
\end{proof}

Second, it is proved that the $\MaxOutQEA$ problem is $\qqQAM$-hard.

\begin{lemma}
$\MaxOutQEA$ is hard for $\qqQAM$ under polynomial-time many-one reduction.
\label{Lemma: qq-QAM-hardness of MaxOutQEA}
\end{lemma}

\begin{proof}
The claim is proved by modifying a part of the proof of Lemma~\ref{Lemma: cqq-QAM is in qq-QAM}.

Let ${A = (A_\yes, A_\no)}$ be a problem in $\qqQAM$,
and let $\function{p}{\Nonnegative}{\Natural}$ be a non-decreasing polynomially bounded function to be specified later.
First notice that $A$ has a $\mathrm{qq}$-QAM proof system
with completeness~${1 - 2^{-p}}$ and soundness~$2^{-p}$. 
Starting from this $\mathrm{qq}$-QAM proof system,
the proof of Lemma~\ref{Lemma: qq-QAM-hardness of CITM}
implies the existence of a polynomial-time algorithm
that, given $x$, computes a description of a quantum circuit~$Q_x$
of ${q_\inp(\abs{x})}$~input qubits and ${q_\out(\abs{x})}$~output qubits
with the following properties:
\begin{itemize}
\item[(i)]
  if ${x}$ is in $A_\yes$,
  there exists a quantum state~$\rho$ consisting of ${q_\inp(\abs{x})}$~qubits such that
  ${D(Q_x(\rho), (I/2)^{\tensor q_\out(\abs{x})}) \leq 2^{-p(\abs{x}) - 1} < 2^{-p(\abs{x})}}$,
  and
\item[(ii)]
  if ${x}$ is in $A_\no$,
  for any quantum state~$\rho$ consisting of ${q_\inp(\abs{x})}$~qubits,
  ${D(Q_x(\rho), (I/2)^{\tensor q_\out(\abs{x})}) > 1/20}$.
\end{itemize}
Let $q$ be another non-decreasing polynomially bounded function satisfying
${q(n) \geq \max \{6, n\}}$ for any $n$ in $\Nonnegative$.
Considering the quantum circuit~$Q'_x$ that runs ${k(\abs{x})}$~copies of $Q_x$ in parallel
for a polynomially bounded function~${k = \bigceil{\frac{2 \ln 2}{\ln (400/399)} q}}$
and taking ${p = q + \ceil{\log k}}$,
it follows from Lemma~\ref{Lemma: polarization of minimum output trace distance} 
that
\begin{itemize}
\item[(i)]
  if $x$ is in $A_\yes$,
  there exists a quantum state~$\rho'$ consisting of ${q'_\inp(\abs{x})}$~qubits
  such that ${D(Q'_x(\rho'), (I/2)^{\tensor q'_\out(\abs{x})}) < 2^{-q(\abs{x})}}$,
  and
\item[(ii)]
  if $x$ is in $A_\no$,
  for any quantum state~$\rho'$ consisting of ${q'_\inp(\abs{x})}$~qubits,
  ${D(Q'_x(\rho'), (I/2)^{\tensor q'_\out(\abs{x})}) > 1 - 2^{-q(\abs{x})}}$,
\end{itemize}
where ${q'_\inp = k q_\inp}$ and ${q'_\out = k q_\out}$.

In what follows, it is assumed that the inequality~${q'_\out(\abs{x}) \leq 2^{q(\abs{x})}}$ holds
(otherwise $\abs{x}$ is at most some fixed constant,
as $q'_\out$ is a polynomially bounded function and ${q(\abs{x}) \geq \abs{x}}$,
and thus,
it can be checked trivially whether $x$ is in $A_\yes$ or in $A_\no$). 
By the second inequality of Lemma~\ref{Lemma: trace distance and entropy},
the circuit~$Q'_x$ satisfies the following properties:
\begin{itemize}
\item[(i)]
  if $x$ is in $A_\yes$,
  there exists a quantum state~$\sigma$ consisting of ${q'_\inp(\abs{x})}$~qubits
  such that 
  ${S(Q'_x(\sigma)) > (1 - 2^{-q(\abs{x})}) q'_\out(\abs{x}) -1 \geq q'_\out(\abs{x}) - 2}$,
and
\item[(ii)]
  if $x$ is in $A_\no$,  
  for any quantum state~$\sigma$ consisting of ${q'_\inp(\abs{x})}$~qubits,
  ${S(Q'_x(\sigma)) < q'_\out(\abs{x}) - q(\abs{x}) + 2 \leq q'_\out(\abs{x}) - 4}$.
\end{itemize}
Thus, 
${(Q'_x, q'_\out(\abs{x})-3)}$ is a yes-instance of $\MaxOutQEA$
if $x$ is in $A_\yes$,
while
it is a no-instance of $\MaxOutQEA$  
if $x$ is in $A_\no$.
This implies that any problem~$A$ in $\qqQAM$ is reducible to $\MaxOutQEA$ in polynomial time,
and the claim follows. 
\end{proof}

Now Theorem~\ref{Theorem: qq-QAM-completeness of MaxOutQEA} follows from Lemmas~\ref{Lemma: MaxOutQEA is in qq-QAM}~and~\ref{Lemma: qq-QAM-hardness of MaxOutQEA}.

\end{document}